\pdfoutput=1
\documentclass[final]{cvpr}

\usepackage{times}
\usepackage{epsfig}
\usepackage{graphicx}
\usepackage{amsmath,amsthm,amscd,amssymb,amsbsy,mathtools}
\usepackage{cancel}
\usepackage[caption=false,subrefformat=parens]{subfig}
\usepackage{epstopdf}
\usepackage{color}
\usepackage{array}
\usepackage{booktabs}
\usepackage{multirow}
\usepackage{makecell}
\usepackage{enumitem}

\usepackage[pagebackref=true,breaklinks=true,colorlinks,bookmarks=false]{hyperref}

\def\cvprPaperID{2803} 
\def\confYear{CVPR 2021}

\def\x{\mathbf x}
\def\r{\mathbf r}
\def\y{\mathbf y}
\def\u{\mathbf u}

\def\hx{{\hat{\mathbf x}}}
\def\hr{{\hat{\mathbf r}}}
\def\hy{{\hat{\mathbf y}}}
\def\hz{{\hat{\mathbf z}}}
\def\hp{\hat{p}}

\def\tx{{\tilde{\mathbf x}}}
\def\tmu{\tilde{\mu}}

\def\bbE{\mathbb E}

\def\btheta{{\boldsymbol \theta}}
\def\bphi{{\boldsymbol \phi}}
\def\bvarphi{{\boldsymbol \varphi}}

\DeclareMathOperator*{\sgn}{sgn}
\DeclareMathOperator*{\logistic}{logistic}

\newtheorem{proposition}{Proposition}

\newcommand{\red}[1] {\textcolor[rgb]{1.0,0.0,0.0}{{#1}}}

\newcolumntype{C}[1]{>{\centering}p{#1}}

\graphicspath{{./figures/}}
\begin{document}

\title{Learning Scalable $\ell_\infty$-constrained Near-lossless Image Compression via \\Joint Lossy Image and Residual Compression}

\author{Yuanchao Bai$^{1}$, Xianming Liu$^{1,2}\thanks{Corresponding author}$, Wangmeng Zuo$^{1,2}$, Yaowei Wang$^{1}$, Xiangyang Ji$^{3}$\\
$^1$Peng Cheng Laboratory, $^2$Harbin Institute of Technology, $^3$Tsinghua University\\
{\tt\small \{baiych,wangyw\}@pcl.ac.cn, \{csxm,wmzuo\}@hit.edu.cn, xyji@tsinghua.edu.cn}
}

\maketitle

\begin{abstract}
   We propose a novel joint lossy image and residual compression framework for learning $\ell_\infty$-constrained near-lossless image compression.
   Specifically, we obtain a lossy reconstruction of the raw image through lossy image compression and uniformly quantize the corresponding residual to satisfy a given tight $\ell_\infty$ error bound.
   Suppose that the error bound is zero, \textit{i.e.}, lossless image compression, we formulate the joint optimization problem of compressing both the lossy image and the original residual in terms of variational auto-encoders and solve it with end-to-end training.
   To achieve scalable compression with the error bound larger than zero, we derive the probability model of the quantized residual by quantizing the learned probability model of the original residual, instead of training multiple networks.
   We further correct the bias of the derived probability model caused by the context mismatch between training and inference.
   Finally, the quantized residual is encoded according to the bias-corrected probability model and is concatenated with the bitstream of the compressed lossy image.
   Experimental results demonstrate that our near-lossless codec achieves the state-of-the-art performance for lossless and near-lossless image compression, and achieves competitive PSNR while much smaller $\ell_\infty$ error compared with lossy image codecs at high bit rates.
\end{abstract}

\section{Introduction}
\label{sec:intro}
Image compression is a ubiquitous technique in computer vision.
For certain applications with stringent demands on image fidelity, such as medical imaging or image archiving, the most reliable choice is lossless image compression. However, the compression ratio of lossless compression is upper-bounded by Shannon's source coding theorem \cite{shannon1948mathematical}, and is typically between $2$:$1$ and $3$:$1$ for practical lossless image codecs \cite{calic,weinberger2000loco,skodras2001j2k,webp,bpg,sneyers2016flif}. To improve the compression performance while keeping the reliability of the decoded images, $\ell_\infty$-constrained near-lossless image compression is developed \cite{chen1994near,ke1998near,nll1998JEI,weinberger2000loco,Wu2000nll} and standardized in traditional codecs, \eg, JPEG-LS \cite{weinberger2000loco} and CALIC \cite{Wu2000nll}.
Different from lossy image compression with Peak Signal-to-Noise Ratio (PSNR) or Multi-Scale Structural SIMilarity index (MS-SSIM) \cite{wang2004ssim,wang2003msssim} distortion measures, $\ell_\infty$-constrained near-lossless image compression requires the maximum reconstruction error of each pixel to be no larger than a given tight numerical bound.
Lossless image compression is a special case of near-lossless image compression, when the tight error bound is zero.

With the fast development of deep neural networks (DNNs), learning-based lossy image compression \cite{Toderici2016iclr,Balle2017iclr,theis2017iclr,toderici2017full,rippel2017icml,li2018cvpr,Balle2018variational,minnen2018nips,Mentzer2018cvpr,Lee2019Context,Choi2019iccv,cheng2020cvpr,li2020pami,Ma2020pami,high2020nips} has achieved tremendous progress over the last four years.
Most recent methods for lossy image compression adopt variational auto-encoder (VAE) architecture \cite{kingma2013auto,kingma2019introduction} based on transform coding \cite{goyal2001theoretical}, where the rate is modeled by the entropy of quantized latent variables and the reconstruction distortion is measured by PSNR or MS-SSIM. Through end-to-end rate-distortion optimization, the state-of-the-art methods, such as \cite{cheng2020cvpr}, can achieve comparable performance with the lastest compression standard Versatile Video Coding (VVC) \cite{vvc}.
However, the above transform coding scheme cannot be directly employed on near-lossless image compression, because it is difficult for DNN-based transforms to satisfy a tight bound on the maximum reconstruction error of each pixel, even without quantization. Invertible transforms, such as integer discrete flow \cite{max2019nips} or wavelet-like transform \cite{Ma2020pami}, are possible solutions to lossless compression but not to general near-lossless compression.

In this paper, we propose a new joint lossy image and residual compression framework for learning near-lossless image compression, inspired by the traditional ``lossy plus residual'' coding scheme \cite{elnahas1986data,melnychuck1989survey,nll1998JEI}.
Specifically, we obtain a lossy reconstruction of the raw image through lossy image compression and uniformly quantize the corresponding residual to satisfy the given $\ell_\infty$ error bound $\tau$.
Suppose that the error bound $\tau$ is zero, \ie, lossless image compression, we formulate the joint optimization problem of compressing both the lossy image and the original residual in terms of VAEs \cite{kingma2013auto,kingma2019introduction}, and solve it with end-to-end training. Note that our VAE model is novel, different from transform coding based VAEs \cite{Balle2017iclr,theis2017iclr,Balle2018variational} for simply lossy image compression or bits-back coding based VAEs \cite{iclr2019bitback,kingma2019bitswap} for lossless image compression.

To achieve scalable near-lossless compression with error bound $\tau>0$, we derive the probability model of the quantized residual by quantizing the learned probability model of the original residual at $\tau=0$, instead of training multiple networks.
Because residual quantization leads to the context mismatch between training and inference, we further propose a bias correction scheme to correct the bias of the derived probability model.
An arithmetic coder \cite{arithmetic_coding} is adopted to encode the quantized residual according to the bias-corrected probability model. Finally, the near-lossless compressed image is stored including the bitstreams of the encoded lossy image and the quantized residual.

Our main contributions are summarized as follows:
\begin{itemize}[leftmargin=*, topsep=0pt]
    \setlength{\itemsep}{0pt}
    \setlength{\parsep}{0pt}
    \setlength{\parskip}{0pt}
    \item We propose a joint lossy image and residual compression framework to realize learning-based lossless and near-lossless image compression. The framework is interpreted as a VAE model and can be end-to-end optimized.
    \item We realize scalable compression by deriving the probability model of the quantized residual from the learned probability model of the original residual, instead of training multiple networks. A bias correction scheme further improves the compression performance.
    \item Our codec achieves the state-of-the-art performance for lossless and near-lossless image compression, and achieves competitive PSNR while much smaller $\ell_\infty$ error compared with lossy image codecs at high bit rates.
\end{itemize}

\section{Related Work}
\label{sec:related_work}
{\bfseries Learning-based Lossy Image Compression.}
Early learning-based methods \cite{Toderici2016iclr,toderici2017full} for lossy image compression are based on recurrent neural networks (RNNs).
Following \cite{Balle2017iclr,theis2017iclr}, most recent learned methods \cite{rippel2017icml,li2018cvpr,Balle2018variational,minnen2018nips,Mentzer2018cvpr,Lee2019Context,Choi2019iccv,cheng2020cvpr,li2020pami,Ma2020pami,high2020nips} adopt convolutional neural networks (CNNs) and can be interpreted as VAEs \cite{kingma2013auto,kingma2019introduction} based on transform coding \cite{goyal2001theoretical}.
In our joint lossy image and residual compression framework, we take advantage of the advanced transform (network structures), quantization and entropy coding techniques in the existing learning-based methods for lossy image compression.

{\bfseries Learning-based Lossless Image Compression.}
Given the strong connections between lossless compression and unsupervised learning, auto-aggressive models \cite{pixelrnn2016icml,pixelcnn,pixelcnn_pp}, flow models \cite{max2019nips}, bits-back coding based VAEs \cite{iclr2019bitback,kingma2019bitswap} and other specific models \cite{Mentzer2019cvpr,mentzer2020cvpr,Ma2020pami}, are introduced to approximate the true distribution of raw images for entropy coding.
Previous work \cite{mentzer2020cvpr} uses traditional BPG lossy image codec \cite{bpg} to compress raw images and proposes a CNN to further compress residuals, which is a special case of our framework. Beyond \cite{mentzer2020cvpr}, our lossy image compressor and residual compressor are jointly optimized through end-to-end training and are interpretable as a VAE model.

{\bfseries Near-lossless Image Compression.}
Basically, previous methods for near-lossless image compression can be divided into three categories: 1) \textit{pre-quantization}: adjusting raw pixel values to the $\ell_\infty$ error bound, and then compressing the pre-processed images with lossless image compression, \eg, near-lossless WebP \cite{webp}; 2) \textit{predictive coding}: predicting subsequent pixels based on previously encoded pixels, then quantizing predication residuals to satisfy the $\ell_\infty$ error bound, and finally compressing the quantized residuals, \eg, \cite{chen1994near,ke1998near}, near-lossless JPEG-LS \cite{weinberger2000loco} and near-lossless CALIC \cite{Wu2000nll}; 3) \textit{lossy plus residual coding}: similar to 2), but replacing predictive coder with lossy image coder, and both the lossy image and the quantized residual are encoded, \eg, in \cite{nll1998JEI}.

In this paper, we propose a joint lossy image and residual compression framework for learning-based near-lossless image compression, inspired by lossy plus residual coding.
\textit{To the best of our knowledge, we propose the first deep-learning-based near-lossless image codec.}
Recently, CNN-based soft-decoding methods \cite{zhang2019dcc,zhang2020ultra} are proposed to improve the reconstruction performance of near-lossless CALIC \cite{Wu2000nll}. However, these methods should belong to image reconstruction rather than image compression.

\section{Scalable Near-lossless Image Compression}
\label{sec:nll_codec}
\subsection{Overview of Compression Framework}
Given a tight $\ell_\infty$ bound $\tau\in\{0,1,2,\ldots\}$, near-lossless codecs compress a raw image satisfying the following distortion measure:
\begin{equation}
    D_{nll}(\x,\hx)=\|\x-\hx\|_\infty=\max_{i,c} |x_{i,c}-\hat{x}_{i,c}|\le\tau
    \label{eq:l_infty}
\end{equation}
where $\x$ and $\hx$ are a raw image and its near-lossless reconstructed counterpart. $x_{i,c}$ and $\hat{x}_{i,c}$ are the pixels of $\x$ and $\hx$, respectively. $i$ denotes the $i$-th spatial position in 2D raster-scan order, and $c$ denotes the $c$-th channel.

In order to realize the $\ell_\infty$ error bound \eqref{eq:l_infty}, we propose a near-lossless image compression framework by integrating lossy image compression and residual compression. We first obtain a lossy reconstruction $\tx$ of the raw image $\x$ through lossy image compression. Lossy image compression methods, such as traditional \cite{wallace1992jpeg,skodras2001j2k,webp,bpg} or learned methods \cite{li2018cvpr,Balle2018variational,minnen2018nips,Mentzer2018cvpr,Lee2019Context,Choi2019iccv,cheng2020cvpr,li2020pami,Ma2020pami}, can achieve high PSNR results at relatively low bit rates, but it is difficult for these methods to ensure a tight error bound $\tau$ of each pixel in $\tx$. We then compute the residual $\r=\x-\tx$ and suppose that $\r$ is quantized to $\hr$. Let $\hx=\tx+\hr$, the reconstruction error $\x-\hx$ is equivalent to the quantization error $\r-\hr$ of $\r$. Thus, we adopt a uniform residual quantizer whose bin size is $2\tau+1$ and quantized value is \cite{weinberger2000loco,Wu2000nll}:
\begin{equation}
    \hat{r}_{i,c}=\sgn(r_{i,c})(2\tau+1)\lfloor(|r_{i,c}|+\tau)/(2\tau+1)\rfloor
    \label{eq:r_quantization}
\end{equation}
where $\sgn(\cdot)$ denotes the sign function. $r_{i,c}$ and $\hat{r}_{i,c}$ are the elements of $\r$ and $\hr$, respectively.
With \eqref{eq:r_quantization}, we now have $|r_{i,c}-\hat{r}_{i,c}|\le\tau$ for each $\hat{r}_{i,c}$ in $\hr$, satisfying the tight error bound \eqref{eq:l_infty}.
Finally, we encode $\hr$ and concatenate it with the compressed $\tx$, and send them to the decoder.

To compress both the lossy image $\tx$ and the quantized residual $\hr$ effectively leads to a challenging joint optimization problem.
In the following subsections, we first propose an end-to-end trainable VAE model \cite{kingma2013auto,kingma2019introduction} to solve the joint optimization problem with $\tau=0$.
Then, we propose a scalable compression scheme with bias correction for $\tau>0$ based on the learned VAE model.

\subsection{Joint Lossy Image \& Residual Compression}
\label{subsec:lossy_residue}
Before addressing the joint compression of the lossy image $\tx$ and the quantized residual $\hr$ with variable $\tau$, we first solve the special case of near-lossless image compression with $\tau=0$, \ie, lossless image compression.

{\bfseries Problem Formulation.} Assuming that raw images are sampled from an unknown probability distribution $p(\x)$, the compression performance of lossless image compression depends on how well we can approximate $p(\x)$ with an underlying model $p_\btheta(\x)$. We adopt the latent variable model which is formulated by a marginal distribution $p_\btheta(\x)=\int p_\btheta(\x,\y)d\y$. $\y$ is an unobserved latent variable and $\btheta$ denote the parameters of this model.

Since directly learning the marginal distribution $p_\btheta(\x)$ is typically intractable, one alternative way is to optimize the evidence lower bound (ELBO) via VAEs \cite{kingma2013auto,kingma2019introduction}. By introducing an inference model $q_\bphi(\y|\x)$ to approximate the posterior $p_\btheta(\y|\x)$, we can rewrite the logarithm of the marginal likelihood $p_\btheta(\x)$ as:
\begin{equation}
    \log p_\btheta(\x)=
    \underbrace{\bbE_{q_\bphi(\y|\x)}\log\frac{p_\btheta(\x,\y)}{q_\bphi(\y|\x)}}_{ELBO}+\underbrace{\bbE_{q_\bphi(\y|\x)}\log\frac{q_\bphi(\y|\x)}{p_\btheta(\y|\x)}}_{D_{kl}(q_\bphi(\y|\x)||p_\btheta(\y|\x))}
    \label{eq:elbo}
\end{equation}
where $D_{kl}(\cdot||\cdot)$ is the Kullback-Leibler (KL) divergence. $\bphi$ denote the parameters of the inference model $q_\bphi(\y|\x)$. Since $D_{kl}(q_\bphi(\y|\x)||p_\btheta(\y|\x))\ge0$ and $\log p_\btheta(\x)\le0$, ELBO is the lower bound of $\log p_\btheta(\x)$.

From \eqref{eq:elbo}, we can minimize the expectation of negative ELBO as a proxy for the expected codelength $\bbE_{p(\x)}[-\log p_\btheta(\x)]$. In our compression framework, we first adopt lossy image compression based on transform coding \cite{goyal2001theoretical}, and the expectation of negative ELBO can be reformulated as follows:
\begin{equation}
    \bbE_{p(\x)}\bbE_{q_\bphi(\hy|\x)}\left[-\log p_\btheta(\x|\hy)-\log p_\btheta(\hy)\right]
    \label{eq:lossy_code}
\end{equation}
where $\hy$ is quantized from continuous $\y$ and $\y$ is transformed from $\x$. Like \cite{Balle2018variational}, we relax the quantization of $\y$ by adding noise from $\mathcal{U}(-\frac12,\frac12)$, and assume $q_\bphi(\hy|\x)=\prod_{i,c} \mathcal{U}(y_{i,c}-\frac12,y_{i,c}+\frac12)$. Thus, $\log q_\bphi(\hy|\x)=0$ is dropped from \eqref{eq:lossy_code}. For simply lossy image compression, such as \cite{Balle2018variational,minnen2018nips,Choi2019iccv,cheng2020cvpr}, the first and second terms of \eqref{eq:lossy_code} are considered as the distortion loss and the rate loss, respectively. Only $\hy$ needs to be encoded.

Beyond lossy image compression, we further take residual compression into consideration.
For each $\x$ and all $(\tx,\r)$ pairs satisfying $\tx+\r=\x$, we have $p_\btheta(\x|\hy)=\sum_{\tx+\r=\x} p_\btheta(\tx,\r|\hy)\ge p_\btheta(\tx,\r|\hy)$.
Hence, we substitute $p_\btheta(\tx,\r|\hy)$ for $p_\btheta(\x|\hy)$, leading to the upper bound of \eqref{eq:lossy_code}.
Let $p_\btheta(\tx,\r|\hy)=p_\btheta(\tx|\hy)\cdot p_\btheta(\r|\tx,\hy)$, we have:
\begin{equation}
    \bbE_{p(\x)}\bbE_{q_\bphi(\hy|\x)}\left[-\cancelto{0}{\log p_\btheta(\tx|\hy)}\underbrace{-\log p_\btheta(\r|\tx,\hy)}_{R_\r}\underbrace{-\log p_\btheta(\hy)}_{R_\hy}\right]
    \label{eq:lossy_residue_code}
\end{equation}
The lossy reconstruction $\tx$ is computed by the inverse transform of $\hy$ and quantization. As the quantization is relaxed by adding uniform noise from $\mathcal{U}(-\frac12,\frac12)$, we have $\log p_\btheta(\tx|\hy)=0$. The second term $R_\r$ and the third term $R_\hy$ of \eqref{eq:lossy_residue_code} are the rates of encoding $\r$ and $\hy$, respectively.

Notice that no distortion loss of lossy image compression is specified in \eqref{eq:lossy_residue_code}. Therefore, we can embed arbitrary lossy image compressors and optimize \eqref{eq:lossy_residue_code} to achieve lossless image compression. A special case is the previous lossless image compression method \cite{mentzer2020cvpr}, in which the BPG lossy image compressor \cite{bpg} with a learned quantization parameter classifier optimizes $-\log p_\btheta(\hy)$ and a CNN-based residual compressor optimizes $-\log p_\btheta(\r|\tx)$.

\begin{figure*}[!t]
\centering
\includegraphics[width=0.99\linewidth]{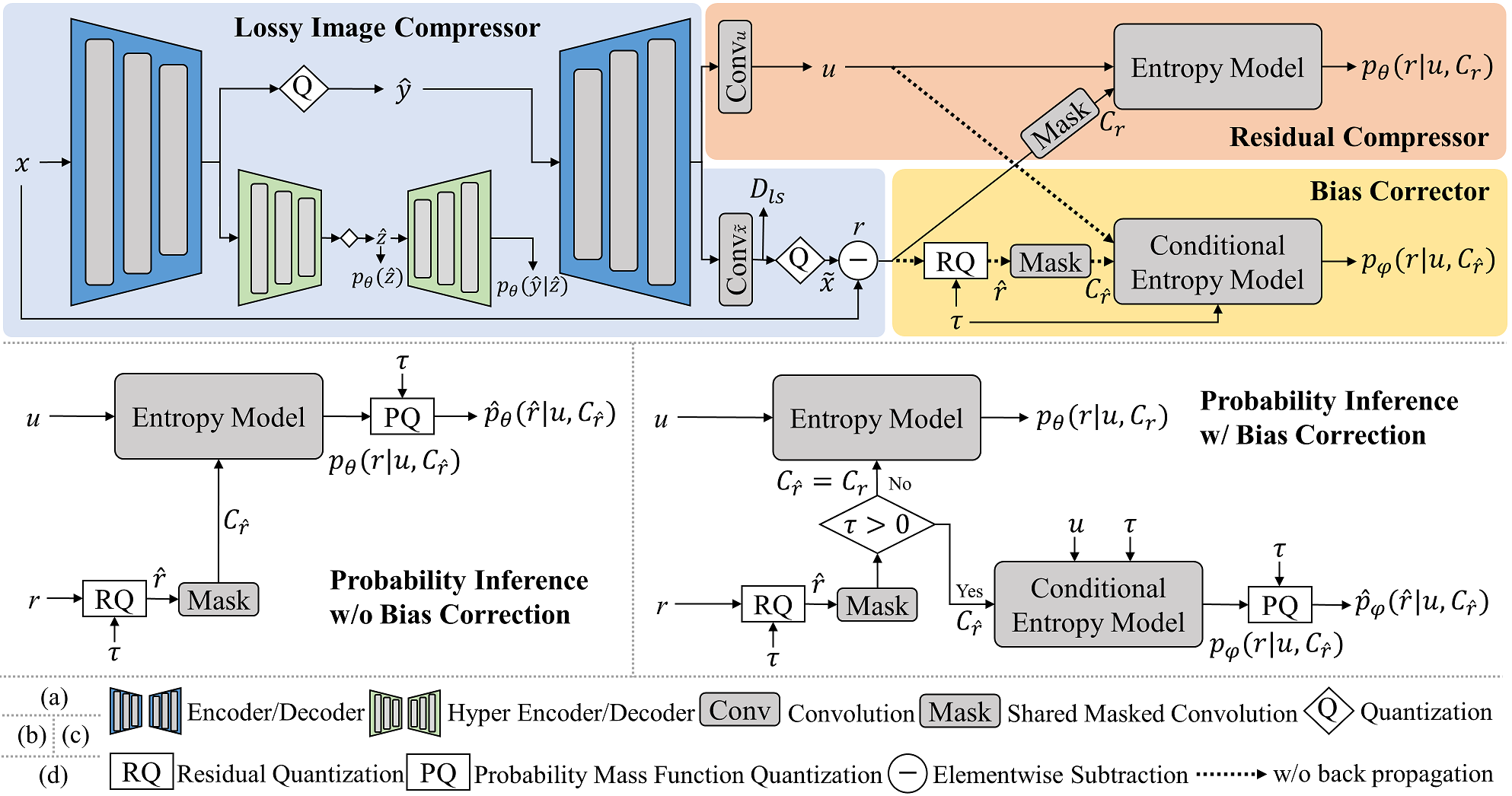}
\caption{(a) Overview of network architecture. (b) Probability inference of the quantized residual $\hr$ without bias correction. (c) Probability inference of the quantized residual $\hr$ with bias correction. (d) Legend with descriptions for notations in (a)-(c).}
\label{fig:network}
\end{figure*}

\begin{figure}[!t]
\centering
\subfloat[]{
\label{fig:em1}
\includegraphics[width=0.49\linewidth]{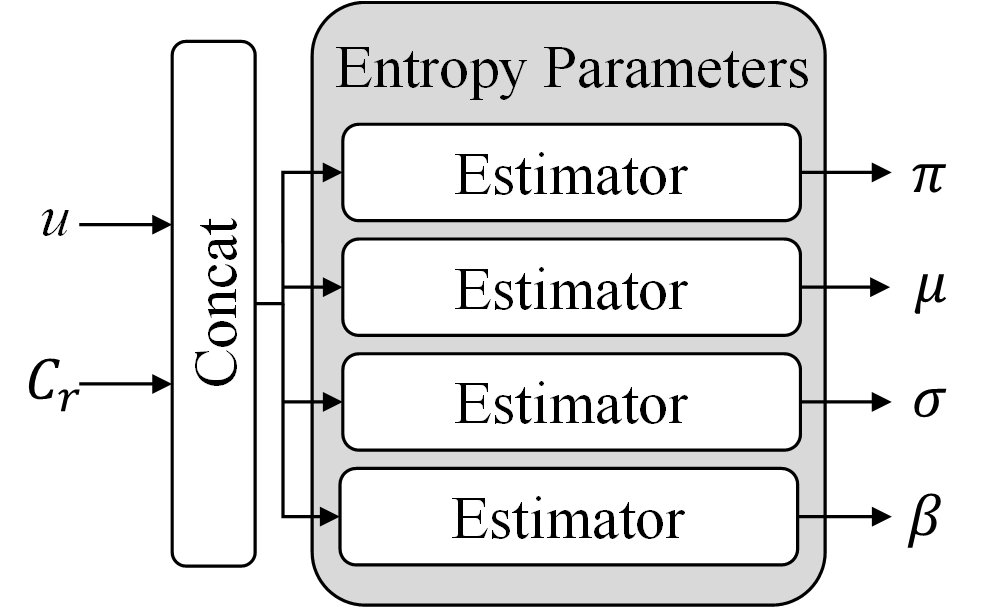}}
\subfloat[]{
\label{fig:em2}
\includegraphics[width=0.49\linewidth]{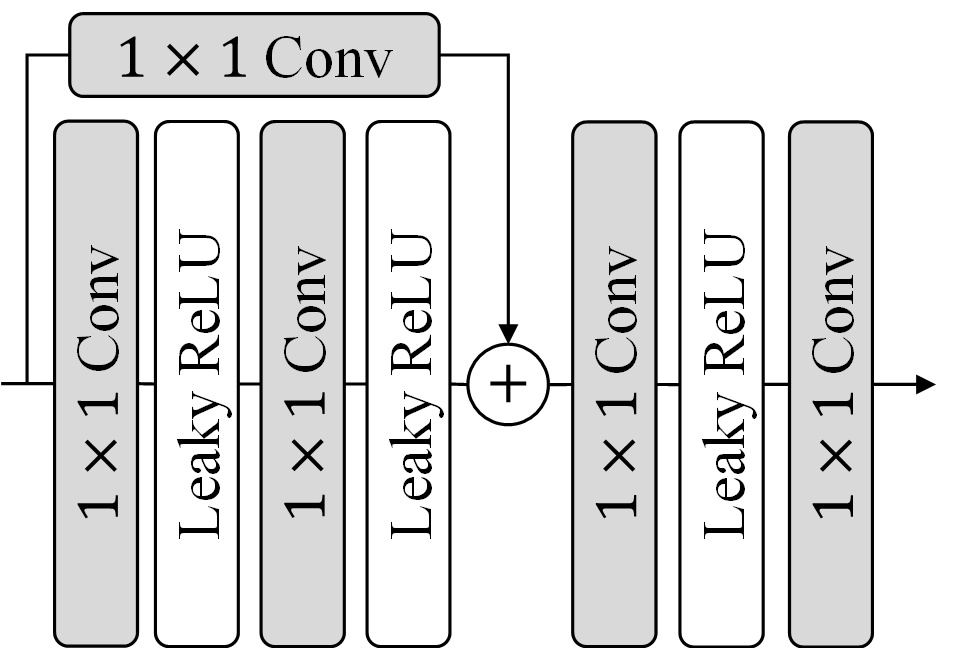}}
\caption{Entropy parameter estimation sub-network. Given $\u$ and $C_\r$, (a) estimates parameters of discrete logistic mixture likelihoods corresponding to \emph{entropy model} in Fig.\;\ref{fig:network}\red{(a)}. (b) is the detailed structure of \emph{estimator} in (a). All $1\times1$ convolutional layers except the last layer have 128 channels. The last convolutional layer has $3\cdot K$ channels.}
\label{fig:entropy_model}
\end{figure}

{\bfseries Network Design and Optimization.} To realize \eqref{eq:lossy_residue_code}, we propose our network architecture in Fig.\;\ref{fig:network}\red{(a)}. For lossy image compression, we employ the hyper-prior model proposed in \cite{Balle2018variational}, where side information $\hz$ is extracted to model the distribution of $\hy$. We employ the residual and attention blocks similar to \cite{cheng2020cvpr,guo2020cvprw} in the \emph{encoder} and \emph{decoder}. Detailed structure of the lossy image compressor is described in the supplementary material. $R_\hy$ is thus extended by
\begin{equation}
    R_{\hy,\hz}=\bbE_{p(\x)}\bbE_{q_\bphi(\hy,\hz|\x)}\left[-\log p_\btheta(\hy|\hz)-\log p_\btheta(\hz)\right]
    \label{eq:rate_y_z}
\end{equation}
where $R_{\hy,\hz}$ is the cost of encoding both $\hy$ and $\hz$.

Denote by $\u=g_\u(\hy)$, where the feature $\u$ is extracted from $\hy$ by $g_\u(\cdot)$. Denote by $\tx=g_\tx(\hy)$, where the lossy reconstruction $\tx$ is inversely transformed from $\hy$ by $g_\tx(\cdot)$. $g_\u$ and $g_\tx$ share the \emph{decoder}, and only the last convolutional layers are different, as shown in Fig.\;\ref{fig:network}\red{(a)}. We interpret $\u$ as the feature of the residual $\r$ given $\tx$ and $\hy$ ($\u$ shares the same height and width with $\r$ and has 64 channels). As $\r=\x-\tx$, we can compute $\r$ and further consider the causal context $C_\r$ of $\r$ during encoding. $R_\r$ is thus extended by
\begin{equation}
    R_\r=\bbE_{p(\x)}\bbE_{q_\bphi(\hy,\hz|\x)}\left[-\log p_\btheta(\r|\u, C_{\r})\right]
    \label{eq:rate_r}
\end{equation}
where the causal context $C_{r_{i}}\in C_\r$ is shared by $r_{i,c}$ of all channels. $C_{r_{i}}$ is extracted from $r_{<i,c}$ by a $5\times5$ masked convolutional layer \cite{minnen2018nips} with 64 channels.
$r_{<i,c}$ denote the residuals encoded before $r_{i,c}$.
For RGB images with three channels, we have
\begin{equation}
    p_\btheta(\r|\u, C_{\r})=\prod_i p_\btheta(r_{i,1},r_{i,2},r_{i,3}|u_i, C_{r_i})
    \label{eq:r_pmf}
\end{equation}

We further model the probability mass function (PMF) of the residual $\r$ with discrete logistic mixture likelihood \cite{pixelcnn_pp,Mentzer2019cvpr,mentzer2020cvpr}.
We use a channel autoregression over $r_{i,1},r_{i,2},r_{i,3}$ and reformulate $p_\btheta(r_{i,1},r_{i,2},r_{i,3}|u_i, C_{r_i})$ as
\begin{align}
    p_\btheta(r_{i,1},r_{i,2},r_{i,3}|u_i, C_{r_i})&= p_\btheta(r_{i,1}|u_i, C_{r_i})\cdot  \label{eq:r_chain_rule} \\
    p_\btheta(r_{i,2}|r_{i,1},u_i, C_{r_i}&)\cdot p_\btheta(r_{i,3}|r_{i,1},r_{i,2},u_i, C_{r_i})  \notag
\end{align}
As shown in Fig.\;\ref{fig:entropy_model}, we introduce a sub-network to estimate entropy parameters, including mixture weights $\pi_{i,c}^k$, means $\mu_{i,c}^k$, variances $\sigma_{i,c}^k$ and mixture coefficients $\beta_{i,t}$. We choose a mixture of $K=5$ logistic distributions.
$k$ denotes the index of the $k$-th logistic distribution. $t$ denotes the channel index of $\beta$.
The channel autoregression over $r_{i,1},r_{i,2},r_{i,3}$ is implemented by updating the means using:
\begin{align}
    \tmu_{i,1}^k=\mu_{i,1}^k,\quad \tmu_{i,2}^k=\mu_{i,2}^k+\beta_{i,1}\cdot r_{i,1}, \notag \\
    \tmu_{i,3}^k=\mu_{i,3}^k+\beta_{i,2}\cdot r_{i,1}+\beta_{i,3}\cdot r_{i,2}
    \label{eq:channel_ar}
\end{align}
With $\pi_{i,c}^k$, $\tmu_{i,c}^k$ and $\sigma_{i,c}^k$, we have
\begin{equation}
    p_\btheta(r_{i,c}|r_{i,<c}, u_i, C_{r_i})\sim\sum_{k=1}^{K} \pi_{i,c}^k \logistic(\tmu_{i,c}^k, \sigma_{i,c}^k)
    \label{eq:factorized_pr}
\end{equation}
where $\logistic(\cdot)$ is the logistic distribution. For discrete $r_{i,c}$, we evaluate $p_\btheta(r_{i,c}|r_{i,<c}, u_i, C_{r_i})$ as \cite{pixelcnn_pp,Mentzer2019cvpr,mentzer2020cvpr}:
\begin{equation}
    \sum_{k=1}^{K} \pi_{i,c}^k \left[S\left(\frac{r_{i,c}^{+}-\tmu_{i,c}^k}{\sigma_{i,c}^k}\right)-S\left(\frac{r_{i,c}^{-}-\tmu_{i,c}^k}{\sigma_{i,c}^k}\right)\right]
    \label{eq:r_eval}
\end{equation}
where $S(\cdot)$ denotes the sigmoid function. $r_{i,c}^{+}=r_{i,c}+0.5$ and $r_{i,c}^{-}=r_{i,c}-0.5$.

Besides rate terms $R_{\hy,\hz}$ and $R_\r$, we add a distortion term $D_{ls}(\x,\tx)=\bbE_{p(\x)}\bbE_{i,c}(x_{i,c}-\tilde{x}_{i,c})^2$ to minimize the mean square error (MSE) between raw image $\x$ and lossy reconstruction $\tx$. As discussed in \cite{Balle2018variational}, minimizing MSE loss is equivalent to learning a lossy image compressor that fits residual $\r$ to a zero-mean factorized Gaussian distribution. However, the discrepancy between the real distribution of $\r$ and the Gaussian distribution could be large. Thus, a more sophisticated entropy model is proposed to encode $\r$. The full loss function for learning near-lossless image compression with $\tau=0$, \ie, lossless image compression, is
\begin{equation}
    \mathcal{L}(\btheta, \bphi)=R_{\hy,\hz}+R_{\r}+\lambda\cdot D_{ls}
    \label{eq:loss_func}
\end{equation}
where $\lambda$ controls the ``rate-distortion'' trade-off. When $\lambda=0$, $\tx$ becomes a latent variable without any constraints. It leads to the best lossless compression performance but is not suitable for our near-lossless image compression with $\tau>0$.
We choose $\lambda=0.03$ empirically in our codec.
Experiments with different $\lambda$'s are conducted in Sec.\;\ref{subset:experiment_ablation}.

\subsection{Probability Inference of Quantized Residuals}
\begin{figure}[!t]
\centering
%
\includegraphics[width=0.96\linewidth]{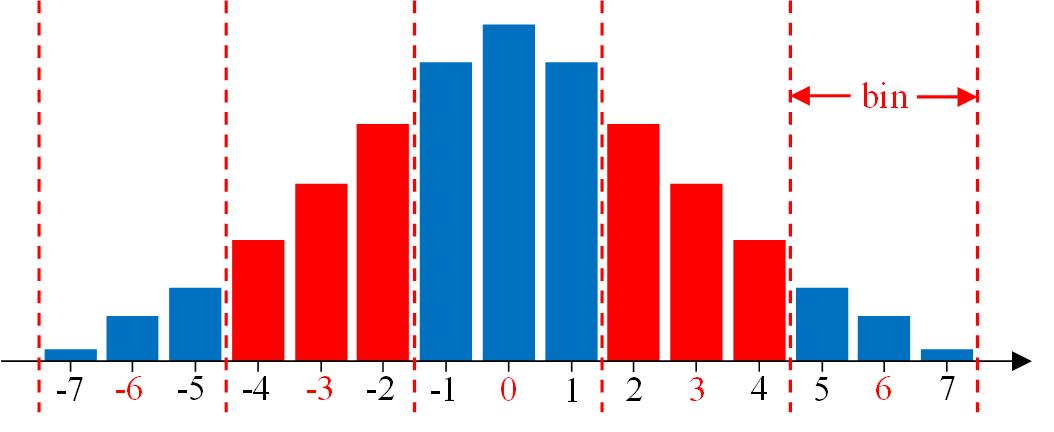}
\caption{PMF quantization corresponding to residual quantization \eqref{eq:r_quantization} with $\tau=1$. Each red number is the value of the quantized residual $\hat{r}_{i,c}$. The probability of each quantized value is the sum of the probabilities of the values in the same bin.}
\label{fig:r_quantization}
\end{figure}

We next propose a scheme to achieve \emph{scalable} near-lossless image compression with $\tau>0$.
For variable $\ell_\infty$ error bound $\tau$, we keep the lossy reconstruction $\tx$ fixed and quantize the residual $\r$ to $\hr$ by \eqref{eq:r_quantization}. Instead of training multiple networks for variably quantized $\hr$, we derive the probability model of $\hr$ from the learned probability model of $\r$ at $\tau=0$. The resulting cost of encoding $\hr$, denoted by $R^\tau_{\hr}$, is reduced significantly with the increase of $\tau$.

Given $\tau$ and the learned PMF $p_\btheta(r_{i,c}|r_{i,<c}, u_i, C_{r_i})$, $\hp_\btheta(\hat{r}_{i,c}|r_{i,<c}, u_i, C_{r_i})$ of quantized $\hat{r}_{i,c}$ can be computed by the following PMF quantization:
\begin{equation}
    \hp_\btheta(\hat{r}_{i,c}|r_{i,<c}, u_i, C_{r_i})=\sum_{\mathclap{v=\hat{r}_{i,c}-\tau}}^{\hat{r}_{i,c}+\tau}p_\btheta(v|r_{i,<c}, u_i, C_{r_i})
    \label{eq:p_rq_ideal}
\end{equation}
An illustrative example is shown in Fig.\;\ref{fig:r_quantization}.
Together with \eqref{eq:r_pmf} and \eqref{eq:r_chain_rule}, we can derive the probability model $\hp_\btheta(\hr|\u,C_\r)$ of $\hr$, which is ideal given the learned $p_\btheta(\r|\u,C_\r)$ of $\r$.

However, encoding $\hr$ with $\hp_\btheta(\hr|\u,C_\r)$ results in undecodable bitstreams, since the original residual $\r$ is unknown to the decoder. $\hp_\btheta(\hat{r}_{i,c}|r_{i,<c}, u_i, C_{r_i})$ cannot be evaluated without $r_{i,<c}$ and causal context $C_{r_i}$.
Instead, we evaluate PMF using the quantized residual $\hr$.
Because of the mismatch between training (with $\r$) and inference (with $\hr$) phases, it leads to a biased PMF $p_\btheta(r_{i,c}|\hat{r}_{i,<c}, u_i, C_{\hat{r}_i})$ of $r_{i,c}$.
We then evaluate $\hp_\btheta(\hat{r}_{i,c}|\hat{r}_{i,<c}, u_i, C_{\hat{r}_i})$ with \eqref{eq:p_rq_ideal} and derive $\hp_\btheta(\hr|\u,C_\hr)$ for the encoding of $\hr$.
The above probability inference scheme is sketched in Fig.\;\ref{fig:network}\red{(b)}.

\subsection{Bias Correction}
\begin{figure}[!t]
\centering
%
\includegraphics[width=0.96\linewidth]{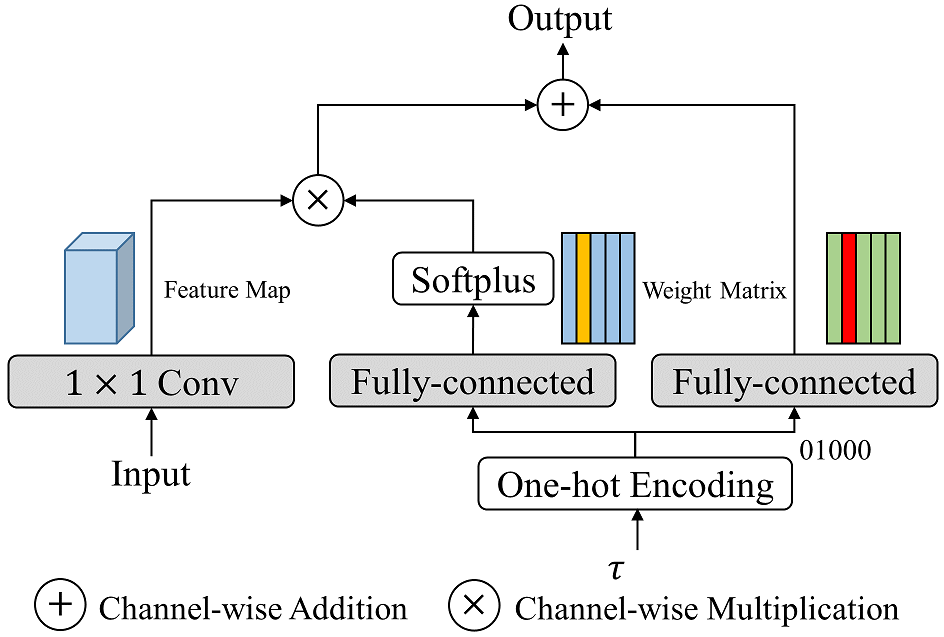}
\caption{Conditional convolutional layer for bias correction. Different outputs can be produced conditioned on $\tau\in\{1,2,\ldots,5\}$.}
\label{fig:conditional_cnn}
\end{figure}

Because of the potential discrepancy between the ideal $\hp_\btheta(\hr|\u,C_\r)$ and the biased $\hp_\btheta(\hr|\u,C_\hr)$, encoding $\hr$ with $\hp_\btheta(\hr|\u,C_\hr)$ degrades the compression performance.
We further propose a bias correction scheme to close the gap between the ideal $\hp_\btheta(\hr|\u,C_\r)$ and the biased $\hp_\btheta(\hr|\u,C_\hr)$, while the resulting bitstreams are still decodable.

The components of the proposed bias corrector are illustrated in Fig.\;\ref{fig:network}\red{(a)}. The masked convolutional layer in the bias corrector is shared with that in the residual compressor. The \emph{conditional entropy model} has the same structure as the \emph{entropy model} illustrated in Fig.\;\ref{fig:entropy_model}, but replaces the convolutional layers with the conditional convolutional layers \cite{pixelcnn,Choi2019iccv} illustrated in Fig.\;\ref{fig:conditional_cnn}.

During training, we generate random $\tau\in\{1,2,\ldots,N\}$ and quantize $\r$ to $\hr$ with \eqref{eq:r_quantization}. We choose $N=5$ in this paper. Given $\u$ and the extracted context $C_\hr$, we use the \emph{conditional entropy model} to estimate $-\log p_\bvarphi(\r|\u, C_{\hr})$ conditioned on different $\tau$, and minimize
\begin{equation}
    \mathcal{L}(\bvarphi)=\bbE_{p(\x)}\bbE_{q_\bphi(\hy,\hz|\x)}\left[\log \frac{p_\btheta(\r|\u, C_{\r}))}{p_\bvarphi(\r|\u, C_{\hr})}\right]
    \label{eq:relative_entropy}
\end{equation}
where $\bvarphi$ denote the parameters of the \emph{conditional entropy model}. $-\log p_\btheta(\r|\u, C_{\r})$ is estimated by the \emph{entropy model} in the residual compressor. $\mathcal{L}(\bvarphi)$ can be considered as an approximate KL-divergence or relative entropy \cite{info_theory} between $p_\btheta(\r|\u, C_{\r})$ and $p_\bvarphi(\r|\u, C_{\hr})$. Because the \emph{entropy model} receives the context $C_\r$ extracted from the original residual $\r$ and is only trained with $\tau=0$, $p_\btheta(\r|\u, C_{\r})$ approximates the true distribution $p(\r|\tx,\hy)$ better than $p_\bvarphi(\r|\u, C_{\hr})$. Thus, $-\log p_\btheta(\r|\u, C_{\r})$ is the lower bound of $-\log p_\bvarphi(\r|\u, C_{\hr})$ on average.

The probability inference scheme with bias correction is sketched in Fig.\;\ref{fig:network}\red{(c)}. For $\tau=0$, we choose the \emph{entropy model} to estimate $p_\btheta(\r|\u, C_{\r})$ to encode $\r$. For $\tau\in\{1,2,\ldots,5\}$, we choose the \emph{conditional entropy model} to estimate $p_\bvarphi(\r|\u, C_{\hr})$ and derive $\hp_\bvarphi(\hr|\u, C_{\hr})$ with \eqref{eq:p_rq_ideal} to encode $\hr$. As $\hp_\bvarphi(\hr|\u, C_{\hr})$ approximates the ideal $\hp_\btheta(\hr|\u,C_\r)$ better than the biased $\hp_\btheta(\hr|\u,C_\hr)$, the compression performance can be improved. Since evaluating $\hp_\bvarphi(\hr|\u, C_{\hr})$ is independent of $\r$, the resulting bitstreams are decodable. We analyze the efficacy of the proposed bias correction scheme in Sec.\;\ref{subset:experiment_ablation}.

{\bfseries Training Strategy.} Bias corrector is trained together with the lossy image compressor and the residual compressor, but minimizing \eqref{eq:relative_entropy} only updates the parameters of the \emph{conditional entropy model} as shown in Fig.\;\ref{fig:network}\red{(a)}. The masked convolutional layer is shared with that in the residual compressor, thus can be updated by minimizing \eqref{eq:loss_func}.
This leads to three advantages: 1) achieving the target \emph{conditional entropy model}; 2) circumventing the relaxation of residual quantization; 3) avoiding degrading the estimation of $p_\btheta(\r|\u, C_{\r})$ caused by training with mixed $\tau$.

\section{Experiments}
\label{sec:experiments}
\subsection{Experimental Setup}
\label{subsec:experiment_setup}
{\bfseries Training.} We train on DIV2K high resolution training dataset \cite{div2k} consisting of 800 2K resolution color images. Although DIV2K is originally built for image super-resolution task, it contains large numbers of high-quality images that is suitable for training our near-lossless codec. The images are first cropped into 27958 samples with the size of $256\times256$.
To augment training data and prevent overfitting on noise, we use bicubic method to downscale each sample with a random factor selected from $[0.6, 1.0]$ and further randomly crop the downscaled sample to $128\times128$ during training.
Our near-lossless codec is optimized for 400 epochs using Adam \cite{kingma2015adam} with minibatches of size $16$. The learning rate is set to $1\times 10^{-4}$  for the first 350 epochs and decays to $1\times 10^{-5}$ for the remaining 50 epochs.

{\bfseries Evaluation.} We evaluate our near-lossless codec on four datasets: 1) Kodak dataset \cite{kodak} consists of 24 uncompressed $768\times512$ color images. 2) DIV2K high resolution validation dataset \cite{div2k} consists of 100 2K resolution color images. 3) CLIC professional validation dataset (CLIC.p) \cite{clic} consists of 41 color images taken by professional photographers. 4) CLIC mobile validation dataset (CLIC.m) \cite{clic} consists of 61 color images taken using mobile phones. Images in CLIC.p and CLIC.m are typically 2K resolution but some of them are of small sizes.

\begin{table}[!tb]
\begin{center}
\small
\begin{tabular}{lC{3.7em}C{3.7em}C{3.7em}C{3.7em}}
\toprule[1pt]
  Codec & Kodak & DIV2K & CLIC.p & CLIC.m \tabularnewline
\hline
\hline
    PNG    & 4.35  & 4.23 & 3.93 & 3.93 \tabularnewline
    JPEG-LS   & 3.16  & 2.99 & 2.82 & 2.53 \tabularnewline
    CALIC   & 3.18  & 3.07 & 2.87 & 2.59 \tabularnewline
    JPEG2000    & 3.19  & 3.12 & 2.93 & 2.71 \tabularnewline
    WebP   & 3.18  & 3.11 & 2.90 & 2.73 \tabularnewline
    BPG    & 3.38  & 3.28 & 3.08 & 2.84 \tabularnewline
    FLIF   & \textbf{2.90}  & 2.91 & 2.72 & \textbf{2.48} \tabularnewline
    L3C    & $-$  & 3.09 & 2.94 & 2.64 \tabularnewline
    RC     & $-$  & 3.08 & 2.93 & 2.54 \tabularnewline
\hline
    Ours   & 3.04  & \textbf{2.81} & \textbf{2.66} & 2.51 \tabularnewline
\bottomrule[1pt]
\end{tabular}
\end{center}
\caption{Compression performance (bpsp) of the proposed near-lossless image codec with $\tau=0$, compared to other lossless image codecs on Kodak, DIV2K, CLIC.p and CLIC.m datasets.}
\label{tb:results_ll}
\end{table}

\begin{table}[!tb]
\begin{center}
\small
\begin{tabular}{lcC{3.4em}C{3.4em}C{3.4em}C{3.4em}}
\toprule[1pt]
  Codec & $\tau^*$ & Kodak & DIV2K & CLIC.p & CLIC.m \tabularnewline
\hline
\hline
    \multirow{3}{*}{\makecell{JPEG- \\LS}}  & 1  & 2.90  & 2.62 & 2.34 & 2.44 \tabularnewline
                                            & 2  & 2.30  & 2.07 & 1.80 & 1.89 \tabularnewline
                                            & 4  & 1.68  & 1.53 & 1.28 & 1.35 \tabularnewline
\hline
    \multirow{3}{*}{\makecell{CALIC}}       & 1  & 2.75  & 2.45 & 2.18 & 2.28 \tabularnewline
                                            & 2  & 2.14  & 1.88 & 1.62 & 1.70 \tabularnewline
                                            & 4  & 1.51  & 1.31 & 1.07 & 1.13 \tabularnewline
\hline
    \multirow{3}{*}{\makecell{WebP \\nll}}  & 1  & 2.41  & 2.45 & 2.26 & 2.11 \tabularnewline
                                            & 2  & 2.01  & 2.04 & 1.89 & 1.85 \tabularnewline
                                            & 4  & 1.82  & 1.83 & 1.73 & 1.75 \tabularnewline
\hline
    \multirow{3}{*}{\makecell{Ours \\w/ bc}}  & 1 & 1.84  & 1.72 & 1.60 & 1.53 \tabularnewline
                                              & 2 & 1.38  & 1.29 & 1.15 & 1.10 \tabularnewline
                                              & 4 & 0.92  & 0.89 & 0.74 & 0.72 \tabularnewline
\bottomrule[1pt]
\end{tabular}
\end{center}
\caption{Compression performance (bpsp) of the proposed near-lossless image codec (w/ bias correction) with $\tau>0$, compared to near-lossless JPEG-LS, near-lossless CALIC and near-lossless WebP on Kodak, DIV2K, CLIC.p and CLIC.m datasets. $^*$\textit{The error bounds of near-lossless WebP are powers of two.}}
\label{tb:results_nll}
\end{table}

\subsection{Performance Evaluation at $\tau=0$}
\label{subsec:experiment_ll}

We first evaluate the compression performance of the proposed near-lossless image codec at $\tau=0$, measured by \textit{bits per subpixel} (bpsp). Each RGB pixel has 3 subpixels. We compare with seven traditional lossless image codecs, \ie, PNG, JPEG-LS \cite{weinberger2000loco}, CALIC \cite{calic}, JPEG2000 \cite{skodras2001j2k}, WebP \cite{webp}, BPG \cite{bpg}, FLIF \cite{sneyers2016flif}, and two the state-of-the-art learned lossless image codecs, \ie, L3C \cite{Mentzer2019cvpr} and RC \cite{mentzer2020cvpr}. Unlike ours, L3C \cite{Mentzer2019cvpr} and RC \cite{mentzer2020cvpr} are trained on a much larger dataset (300,000 images from the Open Images dataset \cite{krasin2017openimages}). We report the compression performance of L3C \cite{Mentzer2019cvpr} and RC \cite{mentzer2020cvpr} published by their authors.

As reported in Table\;\ref{tb:results_ll}, our codec achieves the best performance on DIV2K validation dataset, which shares the same domain with the training dataset.
Our codec achieves the best performance on CLIC.p, and the second best performance on Kodak and CLIC.m, only secondary to FLIF.
The results demonstrate that the jointly learned lossy image compressor and residual compressor with $\tau=0$ can effectively model $p_\btheta(\x)$ and can be generalized to various domains of natural images. This provides a solid foundation for the near-lossless image compression with $\tau>0$.

\subsection{Performance Evaluation at $\tau>0$}
\label{subsec:experiment_nll}

We next evaluate the compression performance of the proposed near-lossless image codec (with bias correction) at $\tau>0$. We compare with near-lossless JPEG-LS \cite{weinberger2000loco}, near-lossless CALIC \cite{Wu2000nll} and near-lossless WebP (WebP nll) \cite{webp}. Besides, we compare with six traditional lossy image codecs, \ie, JPEG \cite{wallace1992jpeg}, JPEG2000 \cite{skodras2001j2k}, WebP \cite{webp}, BPG \cite{bpg}, lossy FLIF \cite{sneyers2016flif} and VVC \cite{vvc}.
Recent learned lossy image codecs published online are trained at relatively low bit rates ($\le 2$ bpp $\approx$ $0.67$ bpsp on Kodak).
We re-implement and train two recent learned lossy image codecs, \ie, Ball\'{e}-hyper \cite{Balle2018variational} and Minnen-joint \cite{minnen2018nips}, at high bit rates ($\ge0.8$ bpsp on Kodak) for comparison.

\begin{figure}[!t]
\centering
\subfloat[]{
\label{fig:rd_error}
\includegraphics[width=0.49\linewidth]{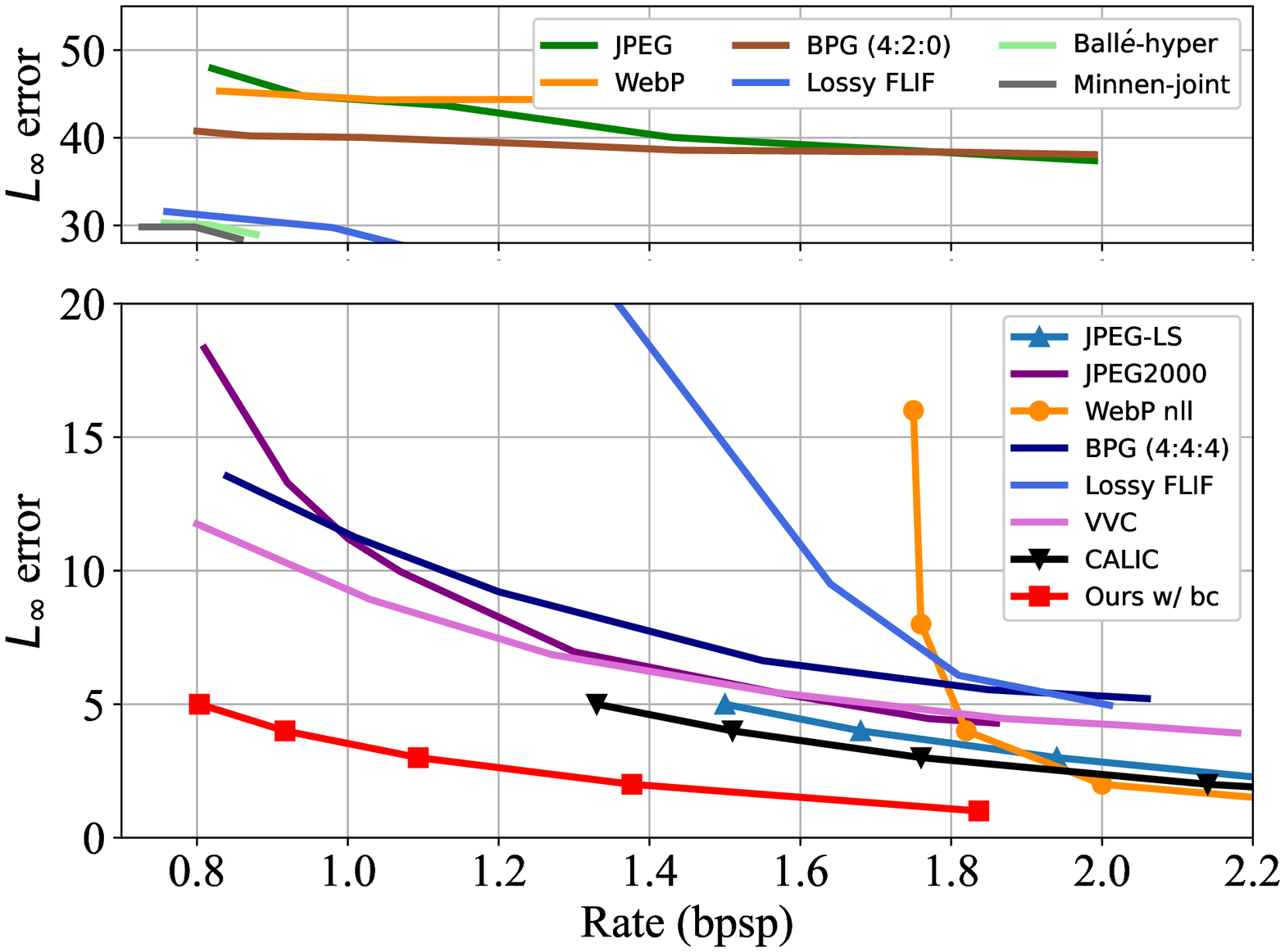}}
\subfloat[]{
\label{fig:rd_psnr}
\includegraphics[width=0.49\linewidth]{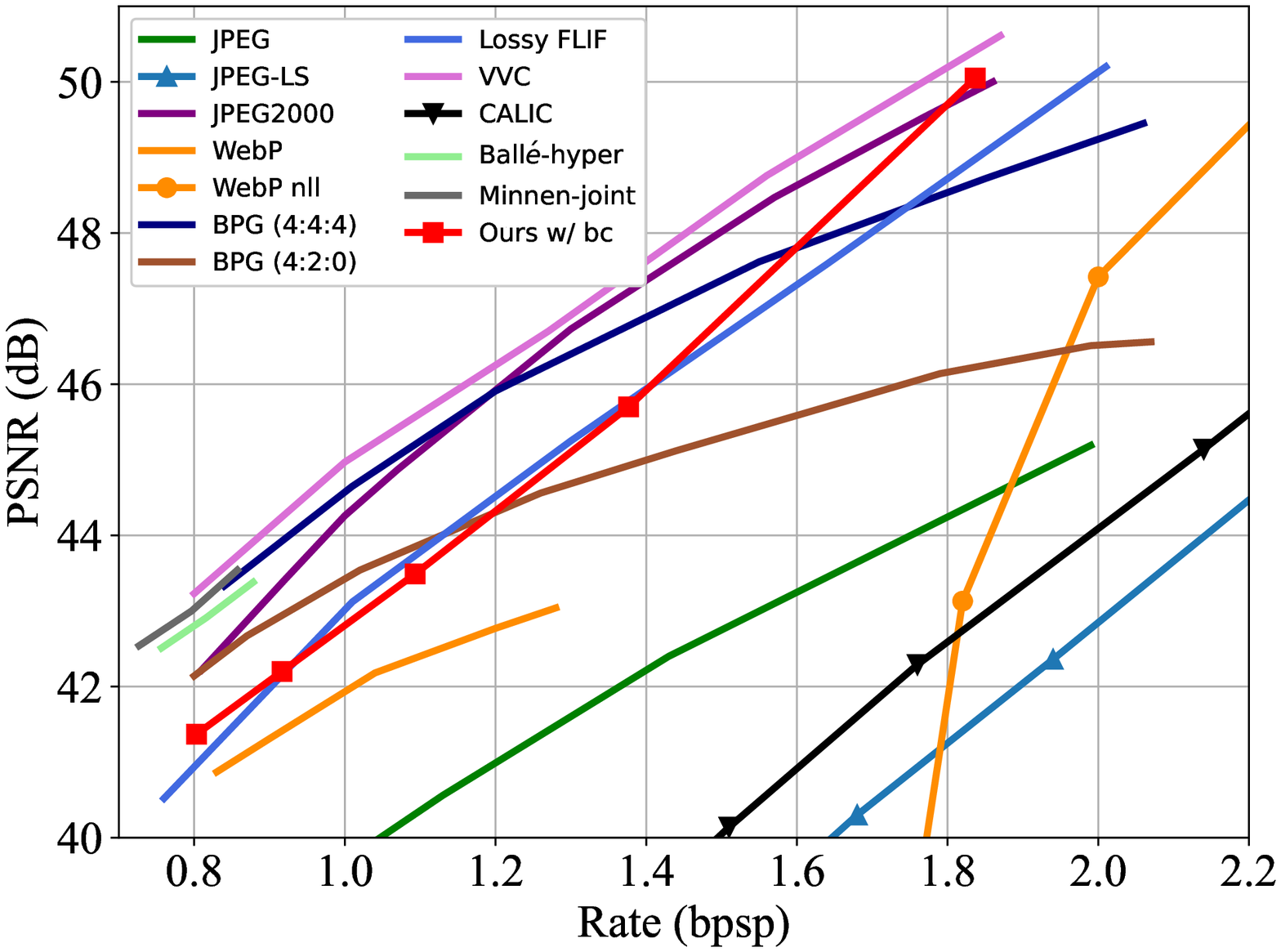}}
\caption{Rate-distortion curves of the proposed near-lossless image codec (w/ bias correction) with $\tau>0$ compared to other near-lossless image codecs and lossy image codecs on Kodak dataset. (a) bpsp-$\ell_\infty$ error curves. (b) bpsp-PSNR curves. }
\label{fig:results_rd}
\end{figure}

\begin{figure}[!t]
\centering
%
\includegraphics[width=0.9\linewidth]{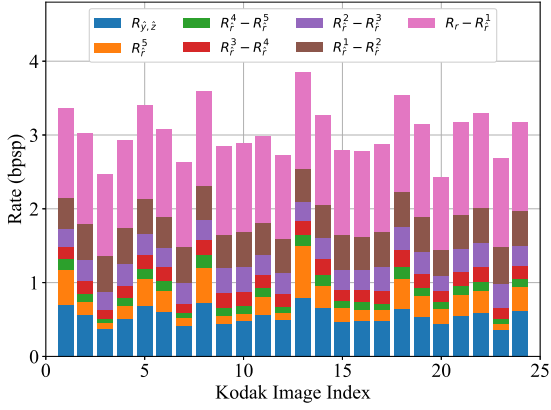}
\caption{Bit rates of each image compressed by the proposed near-lossless image codec with $\tau=\{0,1,\ldots,5\}$ on Kodak dataset.}
\label{fig:rate_ablation}
\end{figure}

Near-lossless WebP simply adjusts pixel values to the $\ell_\infty$ error bound $\tau$ and compresses the pre-processed images losslessly.
Near-lossless JPEG-LS and CALIC adopt predictive coding schemes, and encode the prediction residuals quantized by \eqref{eq:r_quantization}. The predictors and the probability models used in near-lossless JPEG-LS and CALIC are hand-crafted, which are not effective enough.
More effectively, our near-lossless codec is based on jointly trained lossy image compressor and residual compressor. We employ \eqref{eq:r_quantization} to realize the error bound $\tau$ and the probability models of the quantized residuals can be derived from the learned residual compressor with $\tau=0$. Therefore, our codec outperforms near-lossless JPEG-LS, CALIC and WebP by a wide margin, as shown in Table\;\ref{tb:results_nll}.

The rate-distortion performance of all codecs on Kodak dataset is reported in Fig.\;\ref{fig:results_rd}.
In terms of $\ell_\infty$ error defined in \eqref{eq:l_infty}, our codec consistently yields the best results among all codecs, as shown in Fig.\;\subref*{fig:rd_error}.
Besides $\ell_\infty$ error, we also compare PSNR of all codecs, as shown in Fig.\;\subref*{fig:rd_psnr}.
Our codec is better than or comparable with JPEG, near-lossless JPEG-LS, near-lossless CALIC, WebP, near-lossless WebP, BPG (4:2:0) and lossy FLIF.
Our codec is surpassed by JPEG2000, BPG (4:4:4), VVC, Ball\'{e}-hyper and Minnen-joint.
Overall, our codec achieves competitive performance at bit rates higher than $0.8$ bpsp, although PSNR is not our optimization objective.

\subsection{Ablation Study}
\label{subset:experiment_ablation}

{\bfseries Scalability.} In Fig.\;\ref{fig:rate_ablation}, we show the bit rates of each image compressed by the proposed near-lossless image codec with $\tau=\{0,1,\ldots,5\}$ on Kodak dataset. $R_{\hy,\hz}$, on average, accounts for about $18\%$ of $R_{\hy,\hz}+R_\r$ at $\tau=0$. With the increase of $\tau$, the bit rate $R^{\tau}_\hr$ of the quantized residual $\hr$ is significantly reduced.
Especially the $\tau=1$ mode saves about $40\%$ bit rates compared with the $\tau=0$ lossless image compression mode.

\begin{figure}[!t]
\centering
%
\includegraphics[width=0.99\linewidth]{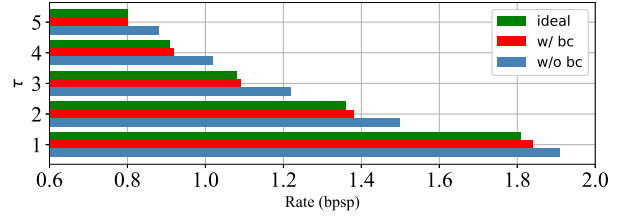}
\caption{Compression performance of the proposed near-lossless image codec without bias correction, with bias correction and with the ideal $\hp_\btheta(\hr|\u,C_\r)$ on Kodak dataset.}
\label{fig:bias_correction}
\end{figure}

\begin{table}[!tb]
\begin{center}
\small
\begin{tabular}{lcC{3em}C{3em}cc}
\toprule[1pt]
  $\lambda$ & Total Rate & $R_{\hy,\hz}$ & $R_\r$ & $\tx$ (PSNR) \tabularnewline
\hline
\hline
     0     & 2.93  & 0.05 & 2.88 & 24.09 \tabularnewline
     0.01  & 3.00  & 0.42 & 2.58 & 37.55 \tabularnewline
     0.03  & 3.04  & 0.55 & 2.49 & 39.29 \tabularnewline
     0.05  & 3.06  & 0.62 & 2.44 & 39.91 \tabularnewline
\bottomrule[1pt]
\end{tabular}
\end{center}
\caption{The effects of different $\lambda$'s on the learned near-lossless image codec with $\tau=0$ on Kodak dataset.}
\label{tb:results_lambda}
\end{table}

\begin{figure*}[!t]
\begin{center}
\subfloat[Raw Image]{
\label{fig:p21_org}
\includegraphics[width=0.195\linewidth]{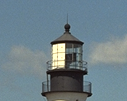}}
\subfloat[lossy $\tx$ ($\lambda=0$)]{
\label{fig:p21_0_ls}
\includegraphics[width=0.195\linewidth]{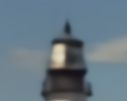}}
\subfloat[lossy $\tx$ ($\lambda=0.03$)]{
\label{fig:p21_3_1s}
\includegraphics[width=0.195\linewidth]{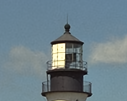}}
\subfloat[nll $\hx$ ($\lambda,\tau=0,5$)]{
\label{fig:p21_0_5}
\includegraphics[width=0.195\linewidth]{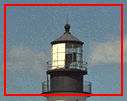}}
\subfloat[nll $\hx$ ($\lambda,\tau=0.03,5$)]{
\label{fig:p21_3_5}
\includegraphics[width=0.195\linewidth]{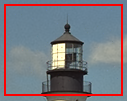}} \\
\subfloat[Raw Image]{
\label{fig:p09_org}
\includegraphics[width=0.195\linewidth]{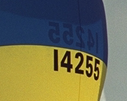}}
\subfloat[lossy $\tx$ ($\lambda=0.01$)]{
\label{fig:p09_1_ls}
\includegraphics[width=0.195\linewidth]{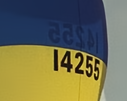}}
\subfloat[lossy $\tx$ ($\lambda=0.03$)]{
\label{fig:p09_3_1s}
\includegraphics[width=0.195\linewidth]{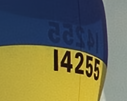}}
\subfloat[nll $\hx$ ($\lambda,\tau=0.01,5$)]{
\label{fig:p09_1_5}
\includegraphics[width=0.195\linewidth]{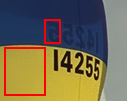}}
\subfloat[nll $\hx$ ($\lambda,\tau=0.03,5$)]{
\label{fig:p09_3_5}
\includegraphics[width=0.195\linewidth]{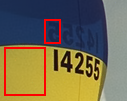}}
\end{center}
\caption{Visualizing the effects of different $\lambda$'s on the lossy reconstruction $\tx$ and the near-lossless (nll) reconstruction $\hx$. \itshape{The images are best viewed in full size on screen}. }
\label{fig:resulting_patches}
\end{figure*}

{\bfseries Bias Correction.}
In Fig.\;\ref{fig:bias_correction}, we demonstrate the efficacy of bias correction at $\tau>0$. Because of the discrepancy between the ideal $\hp_\btheta(\hr|\u,C_\r)$ and the biased $\hp_\btheta(\hr|\u,C_\hr)$, encoding $\hr$ with $\hp_\btheta(\hr|\u,C_\hr)$ (without bias correction) degrades the compression performance.
Instead, we encode $\hr$ with $\hp_\bvarphi(\hr|\u, C_{\hr})$ (with bias correction) resulting in lower bit rates.
With the increase of $\tau$, the PMF of $\r$ is quantized by larger bins and becomes coarser. Thus, the compression performance with bias correction approaches the ``ideal''.
However, the gap between the compression performance without bias correction and the ``ideal'' remains large, as the discrepancy between $\hp_\btheta(\hr|\u,C_\r)$ and $\hp_\btheta(\hr|\u,C_\hr)$ is also magnified with the increasing $\tau$. Note that the compression performance of our codec without bias correction is also better than near-lossless JPEG-LS, CALIC and WebP.

{\bfseries Discussion on different $\lambda$'s.} The $\lambda$ in \eqref{eq:loss_func} adapts the rate of losslessly encoding the raw image $\x$ and the distortion of the lossy reconstruction $\tx$. In Table\;\ref{tb:results_lambda}, we evaluate the effects of different $\lambda\in\{0,0.01,0.03,0.05\}$ on Kodak dataset.
With the increase of $\lambda$, the PSNR of $\tx$ is improved but the bpsp of encoding $\x$ also becomes higher. If we simply do lossless image compression ($\tau=0$), we should choose $\lambda=0$. The $\tx$ is learned by the network without any hand-crafted constraints, leading to the best compression performance.

However, $\lambda=0$ is not suitable for our near-lossless codec with $\tau>0$. As the average PSNR of the $\tx$ is $24.09$ dB, the magnitude of $\r=\x-\tx$ is large. To quantize such large residuals with large $\tau$ results in blocking artifacts in the near-lossless reconstruction $\hx$. As shown in Fig.\;\subref*{fig:p21_0_ls}, $\tx$ with $\lambda=0$ is blurry. If $\tau=5$, the quantization bin size is large, \ie, $2\tau+1=11$, resulting in blocking artifacts in $\hx$ in Fig.\;\subref*{fig:p21_0_5}. For $\lambda=0.01$, we also observe small artifacts in $\hx$ caused by quantization with large $\tau$, as shown in Fig.\;\subref*{fig:p09_1_5}. In our codec, we choose $\lambda=0.03$, which is the best trade-off between compression performance and visual quality in the experiments. As shown in Fig.\;\subref*{fig:p21_3_5} and Fig.\;\subref*{fig:p09_3_5}, we observe no artifacts in $\hx$ compared with $\lambda=0,\;0.01$.

\section{Conclusion}
\label{sec:conclusion}
In this paper, we propose a joint lossy image and residual compression framework to learn a scalable $\ell_\infty$-constrained near-lossless image codec.
With our codec, a raw image is first compressed by the lossy image compressor, and the corresponding residual is then quantized to satisfy variable $\ell_\infty$ error bounds.
To compress the variably quantized residuals, the probability models of the quantized residuals are derived from the learned probability model of the original residual, which is realized by the residual compressor in conjunction with the bias corrector.
Experimental results demonstrate the state-of-the-art performance of our near-lossless image codec.

\section*{Acknowledgement}
\small{This work was supported by National Natural Science Foundation of China under Grants 61922027, 61827804 and U20B2052, National Key Research and Development Project under Grant 2019YFE0109600, China Postdoctoral Science Foundation under Grant 2020M682826.}

\onecolumn{
\appendix
\section{Lossy Image Compressor Architecture}
As shown in Fig.\;\ref{fig:lossy_image_compressor}, we employ the hyper-prior model proposed in \cite{Balle2018variational}, where side information $\hz$ is extracted to model the distribution of $\hy$. We employ the residual and attention blocks to model analysis/synthesis transforms, similar to \cite{cheng2020cvpr,guo2020cvprw}.

\begin{figure}[htb]
\centering
\includegraphics[width=0.95\linewidth]{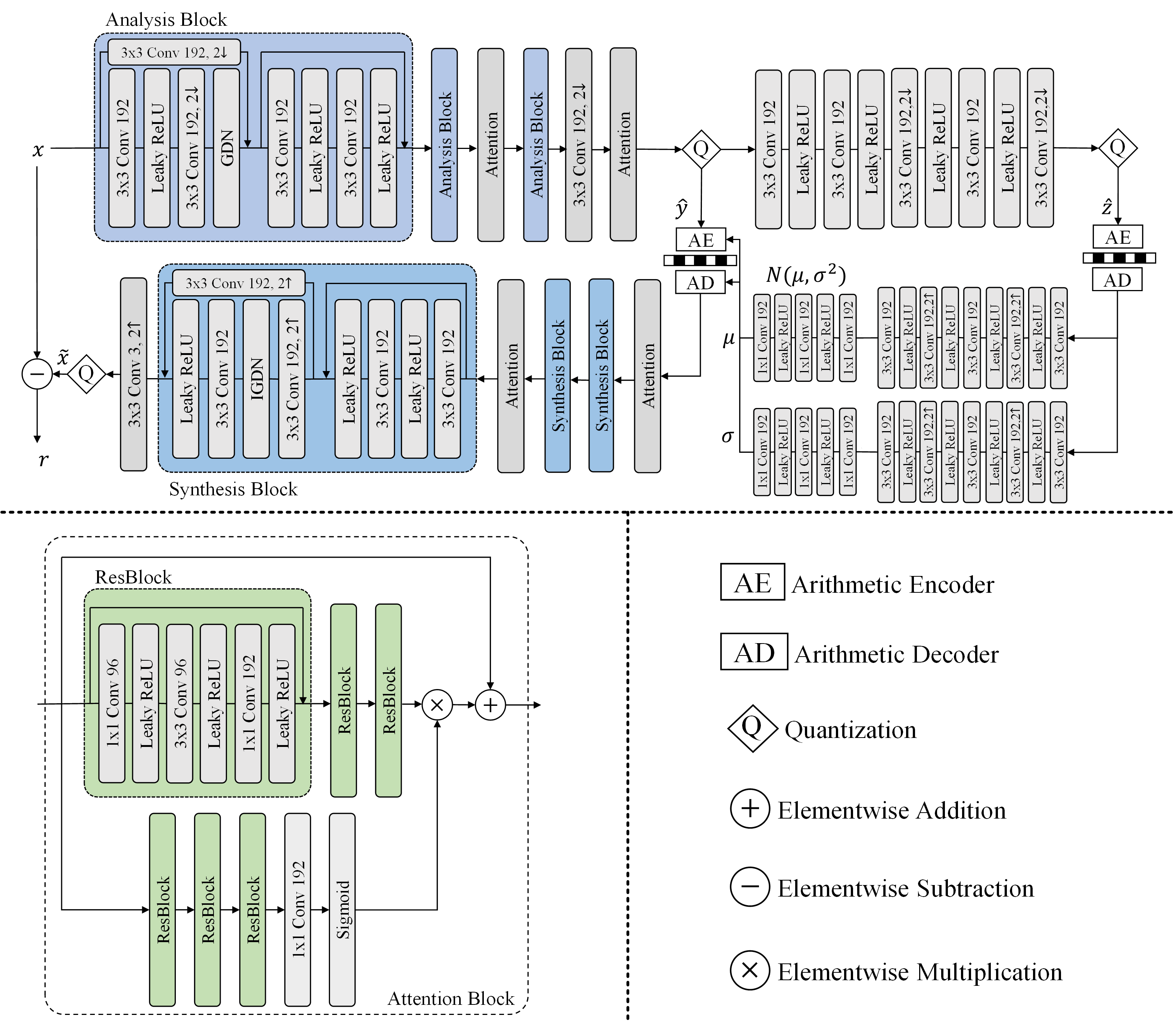}
\caption{Lossy image compressor architecture. Top: overview. Bottom left: attention block. Bottom right: legend.}
\label{fig:lossy_image_compressor}
\end{figure}

\section{Proof of Entropy Reduction Resulting from Uniform Quantization}
\begin{proposition}
    Let $x\in\{a,a+1,\ldots,a+N\}$ be a discrete random variable with probability mass function (PMF) $p(x)$, where $a$ and $N\ge0$ are two integers. Suppose that $x$ is quantized to $\hat{x}$ by
    \begin{equation}
        \hat{x}=\sgn(x)(2\tau+1)\lfloor(|x|+\tau)/(2\tau+1)\rfloor,\;\; \tau=0,1,\ldots,M
        \label{eq:x_quantization}
    \end{equation}
    where $\sgn(\cdot)$ denotes the sign function and $M\ge0$ is an integer. The PMF $\hp(\hat{x})$ of the quantized $\hat{x}$ is
    \begin{equation}
        \hp(\hat{x})=\sum_{x=\hat{x}-\tau}^{\hat{x}+\tau}p(x)
        \label{eq:pmf_quantization}
    \end{equation}
    Let $H(x)$ and $H(\hat{x})$ be the entropy of $x$ and $\hat{x}$ respectively, we have
     \begin{equation}
        H(x)\ge H(\hat{x})
    \end{equation}
    In words, uniformly quantizing $x$ with \eqref{eq:x_quantization} leads to entropy reduction.
    \label{prop:quantization}
\end{proposition}
\begin{proof}
    Based on the definition of entropy and \eqref{eq:pmf_quantization}, we have
    \begin{align}
        H(x)&=-\sum_x p(x)\log p(x)\\
        &=-\sum_{\hat{x}}\sum_{x=\hat{x}-\tau}^{\hat{x}+\tau}p(x)\log p(x) \\
        &\ge-\sum_{\hat{x}}\sum_{x=\hat{x}-\tau}^{\hat{x}+\tau}p(x)\log \hp(\hat{x}) \label{eq:ge_proof}\\
        &=-\sum_{\hat{x}}\hp(\hat{x})\log \hp(\hat{x})= H(\hat{x})
    \end{align}
    Inequation \eqref{eq:ge_proof} holds, because $\hp(\hat{x})\ge p(x)$ for $x\in[\hat{x}-\tau,\hat{x}+\tau]$.
\end{proof}
Proposition\;\ref{prop:quantization} is the theoretical foundation for our scalable compression scheme corresponding to residual quantization.

\section{More Visual Results}
We show more visual results of the near-lossless reconstructions of our codec on Kodak dataset \cite{kodak} in Fig.\;\ref{fig:kodim09}, \ref{fig:kodim13}, \ref{fig:kodim15} and \ref{fig:kodim21}. Compared with the raw images, no visual artifacts are introduced by residual quantization. Human eyes can hardly differentiate between the the raw images and the near-lossless reconstructions of our codec at $\tau\le5$.

\begin{figure}[!t]
\begin{center}
\subfloat[Raw image / $\lambda,\tau=0.03,0$]{
\label{fig:k09_0}
\includegraphics[width=0.3\linewidth]{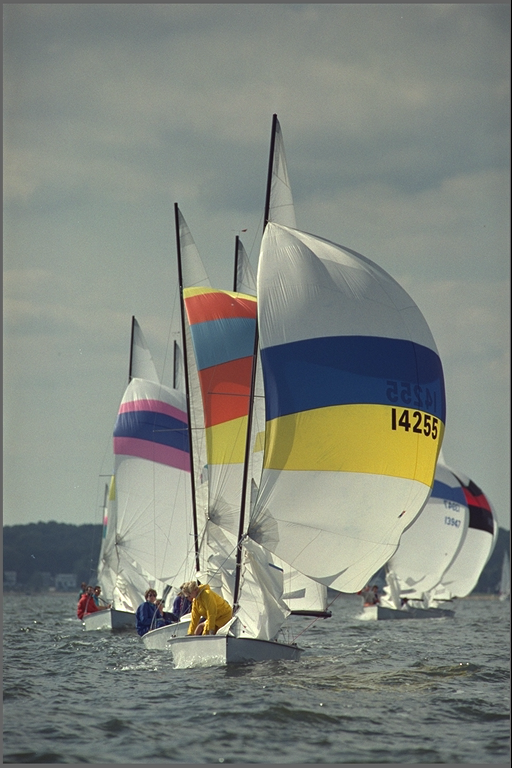}}
\subfloat[$\lambda,\tau=0.03,1$]{
\label{fig:k09_1}
\includegraphics[width=0.3\linewidth]{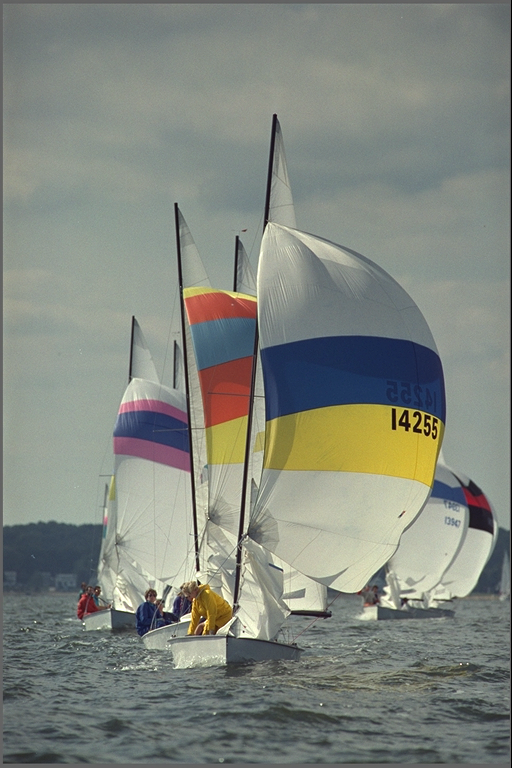}}
\subfloat[$\lambda,\tau=0.03,2$]{
\label{fig:k09_2}
\includegraphics[width=0.3\linewidth]{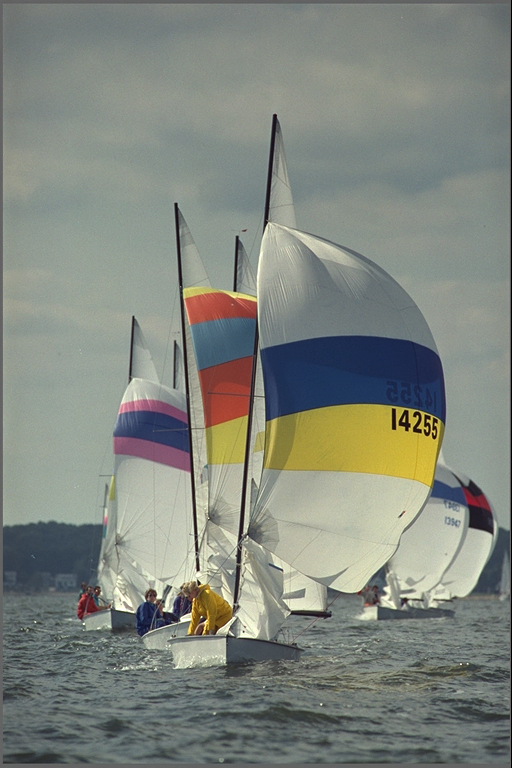}}\\
\subfloat[$\lambda,\tau=0.03,3$]{
\label{fig:k09_3}
\includegraphics[width=0.3\linewidth]{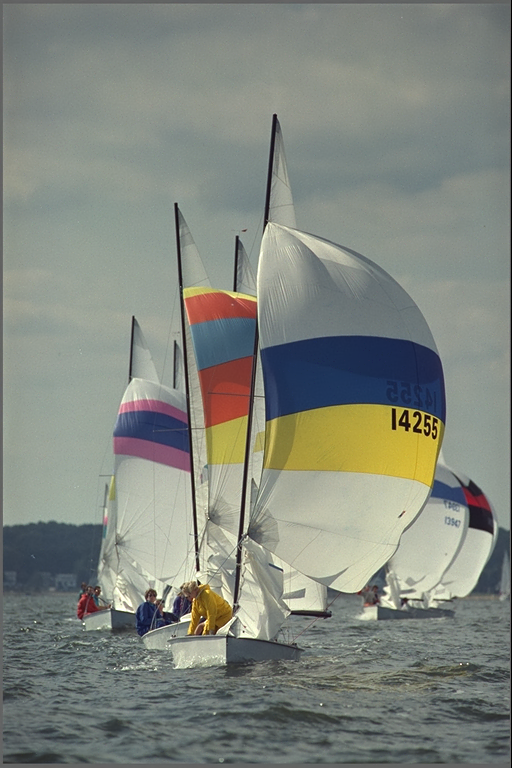}}
\subfloat[$\lambda,\tau=0.03,4$]{
\label{fig:k09_4}
\includegraphics[width=0.3\linewidth]{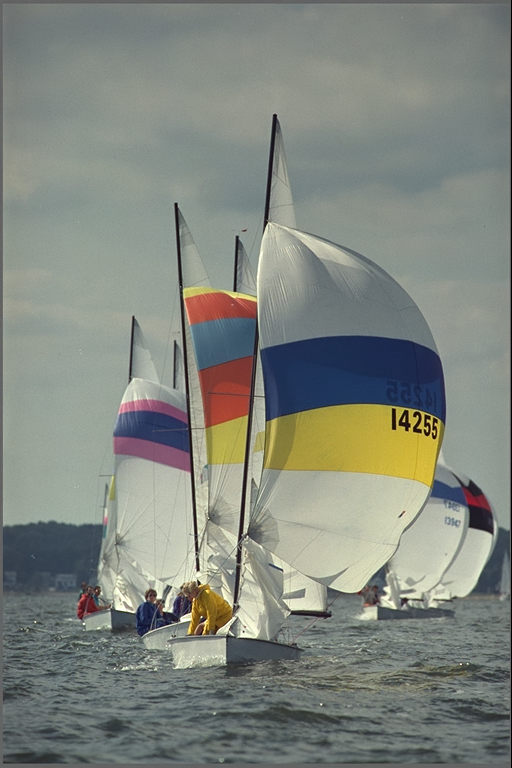}}
\subfloat[$\lambda,\tau=0.03,5$]{
\label{fig:k09_5}
\includegraphics[width=0.3\linewidth]{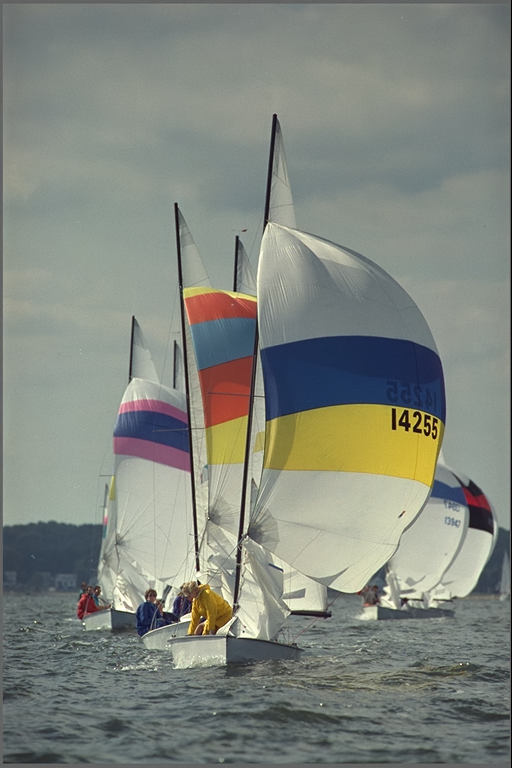}}
\end{center}
\caption{\textit{Kodim09}. Near-lossless image codec is trained with $\lambda=0.03$. Visualization of near-lossless reconstructions with different $\tau$.}
\label{fig:kodim09}
\end{figure}

\begin{figure}[!t]
\begin{center}
\subfloat[Raw image / $\lambda,\tau=0.03,0$]{
\label{fig:k13_0}
\includegraphics[width=0.45\linewidth]{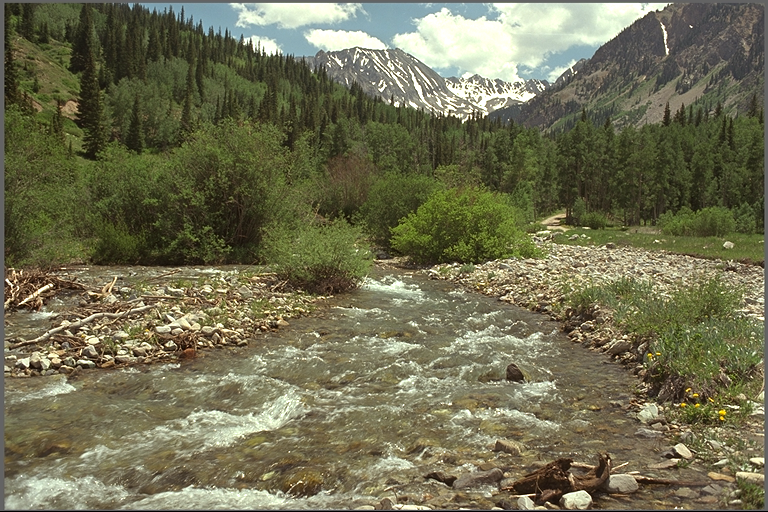}}
\subfloat[$\lambda,\tau=0.03,1$]{
\label{fig:k13_1}
\includegraphics[width=0.45\linewidth]{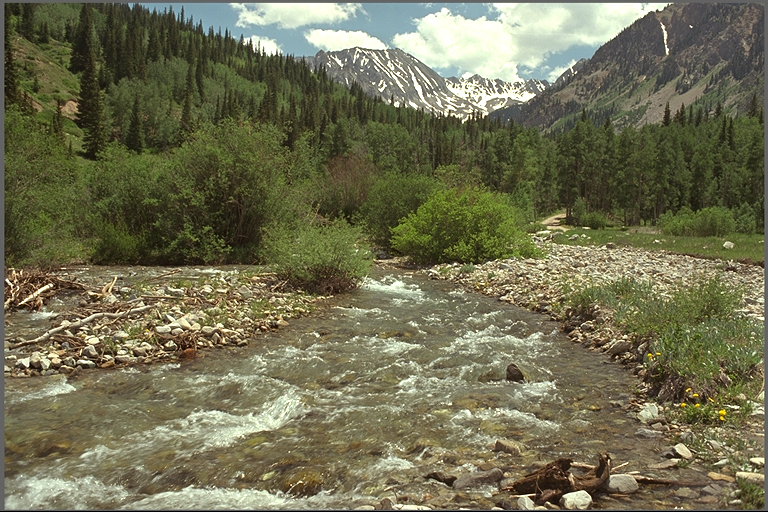}}\\
\subfloat[$\lambda,\tau=0.03,2$]{
\label{fig:k13_2}
\includegraphics[width=0.45\linewidth]{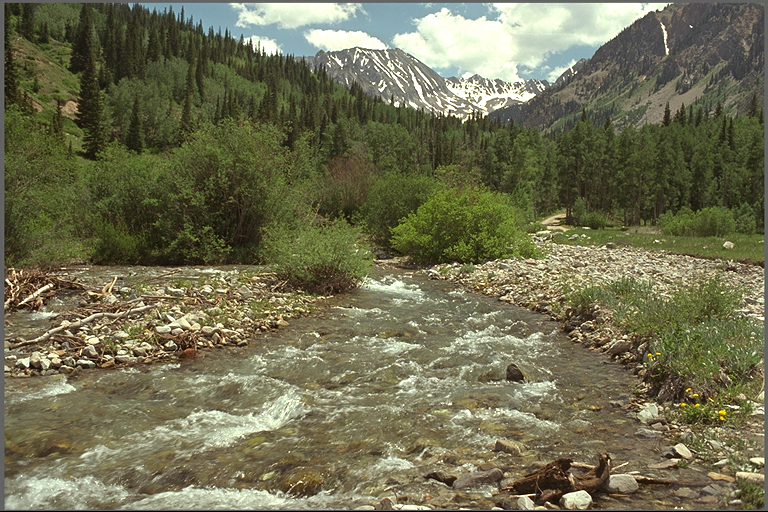}}
\subfloat[$\lambda,\tau=0.03,3$]{
\label{fig:k13_3}
\includegraphics[width=0.45\linewidth]{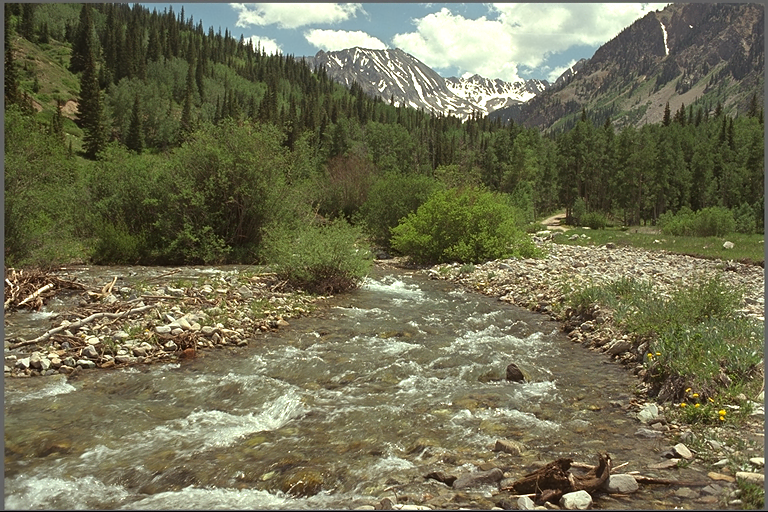}}\\
\subfloat[$\lambda,\tau=0.03,4$]{
\label{fig:k13_4}
\includegraphics[width=0.45\linewidth]{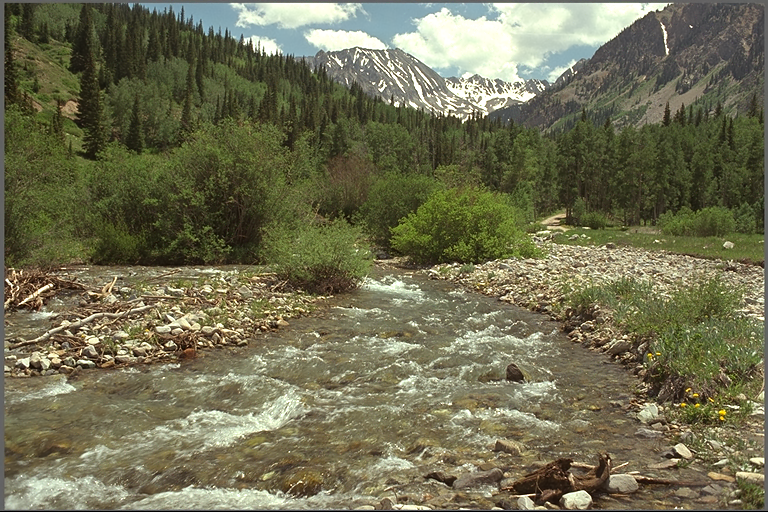}}
\subfloat[$\lambda,\tau=0.03,5$]{
\label{fig:k13_5}
\includegraphics[width=0.45\linewidth]{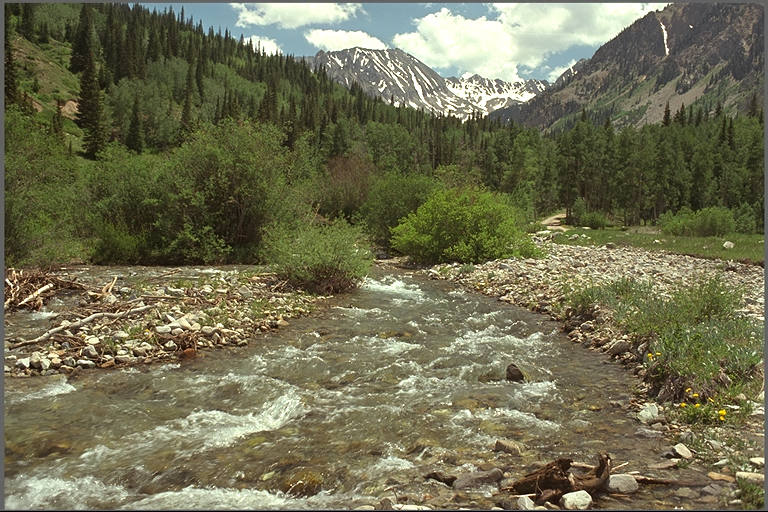}}
\end{center}
\caption{\textit{Kodim13}. Near-lossless image codec is trained with $\lambda=0.03$. Visualization of near-lossless reconstructions with different $\tau$.}
\label{fig:kodim13}
\end{figure}

\begin{figure}[!t]
\begin{center}
\subfloat[Raw image / $\lambda,\tau=0.03,0$]{
\label{fig:k15_0}
\includegraphics[width=0.45\linewidth]{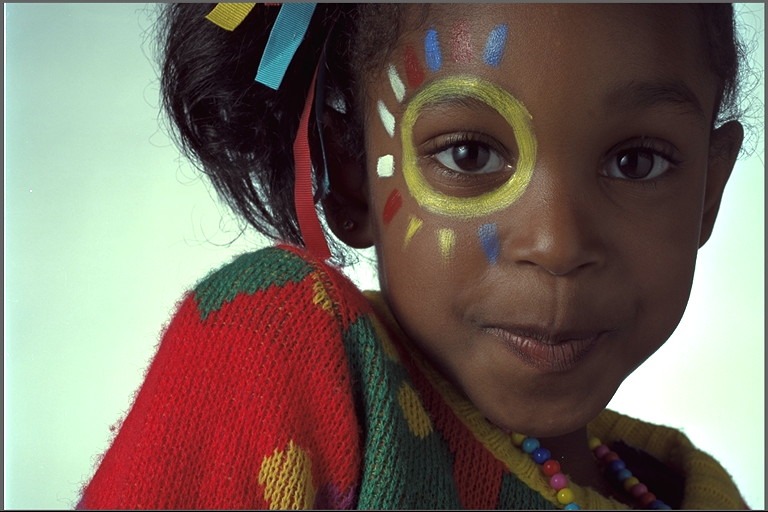}}
\subfloat[$\lambda,\tau=0.03,1$]{
\label{fig:k15_1}
\includegraphics[width=0.45\linewidth]{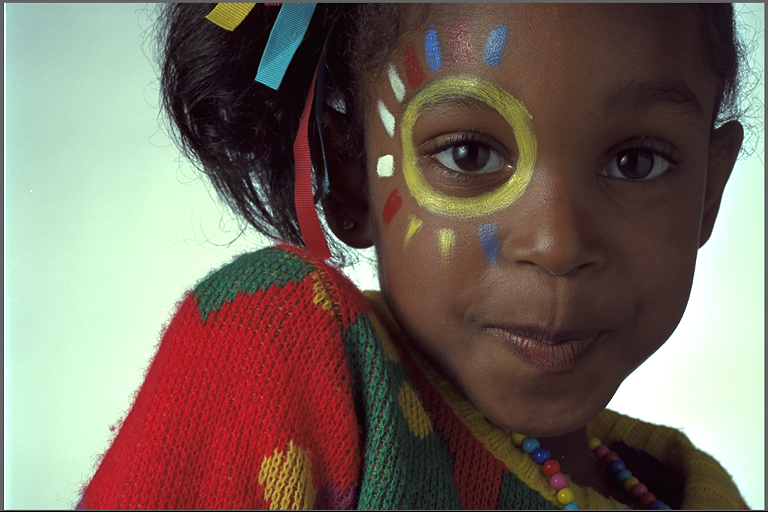}}\\
\subfloat[$\lambda,\tau=0.03,2$]{
\label{fig:k15_2}
\includegraphics[width=0.45\linewidth]{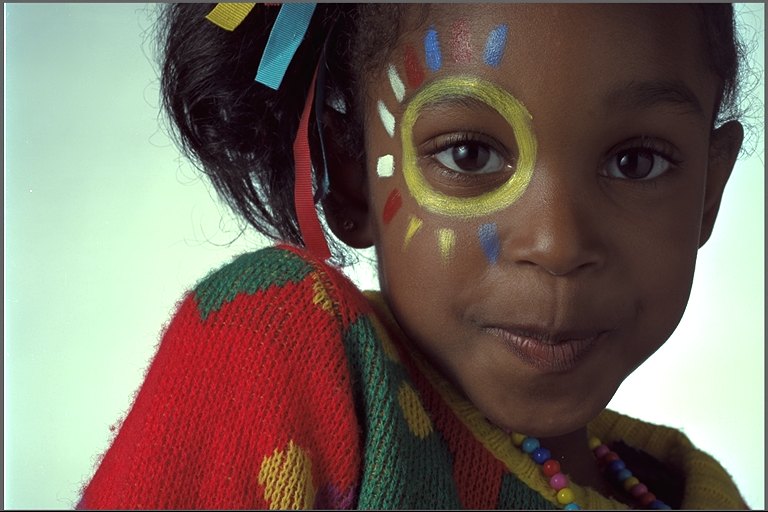}}
\subfloat[$\lambda,\tau=0.03,3$]{
\label{fig:k15_3}
\includegraphics[width=0.45\linewidth]{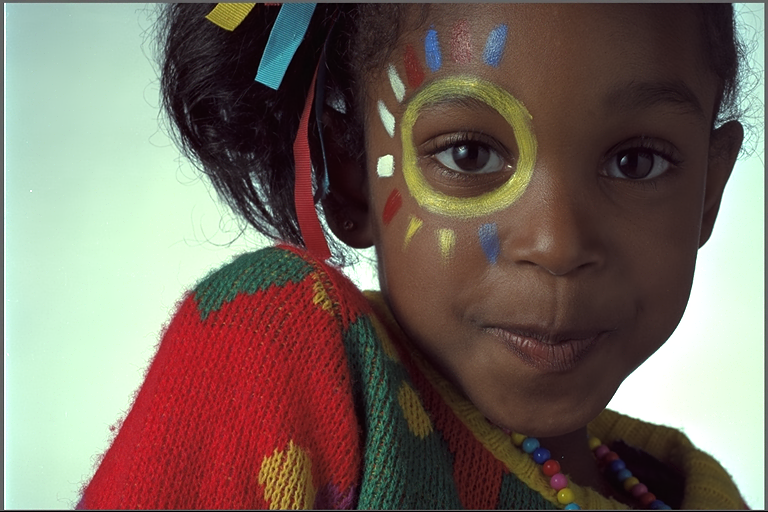}}\\
\subfloat[$\lambda,\tau=0.03,4$]{
\label{fig:k15_4}
\includegraphics[width=0.45\linewidth]{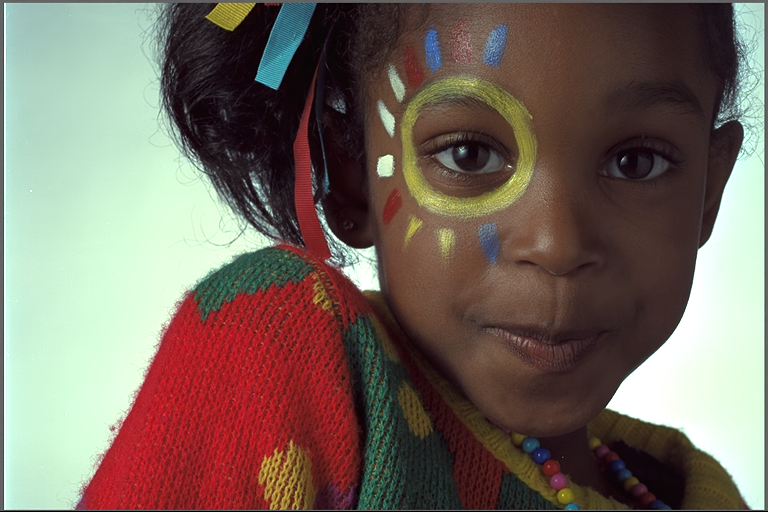}}
\subfloat[$\lambda,\tau=0.03,5$]{
\label{fig:k15_5}
\includegraphics[width=0.45\linewidth]{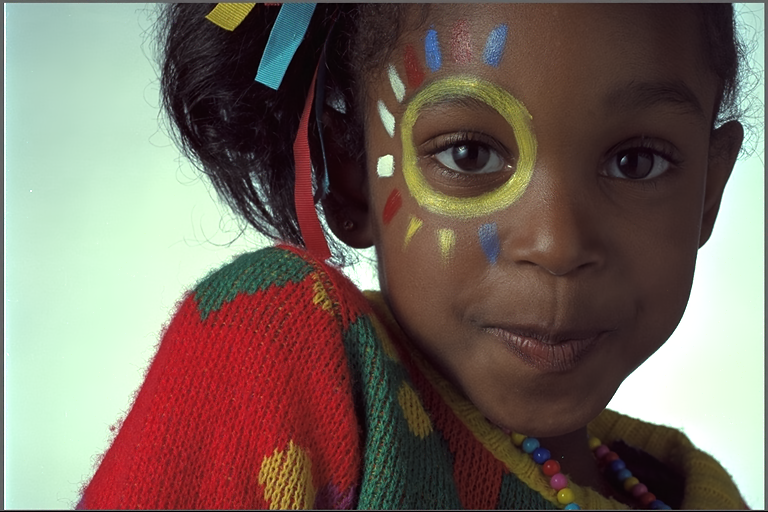}}
\end{center}
\caption{\textit{Kodim15}. Near-lossless image codec is trained with $\lambda=0.03$. Visualization of near-lossless reconstructions with different $\tau$.}
\label{fig:kodim15}
\end{figure}

\begin{figure}[!t]
\begin{center}
\subfloat[Raw image / $\lambda,\tau=0.03,0$]{
\label{fig:k21_0}
\includegraphics[width=0.45\linewidth]{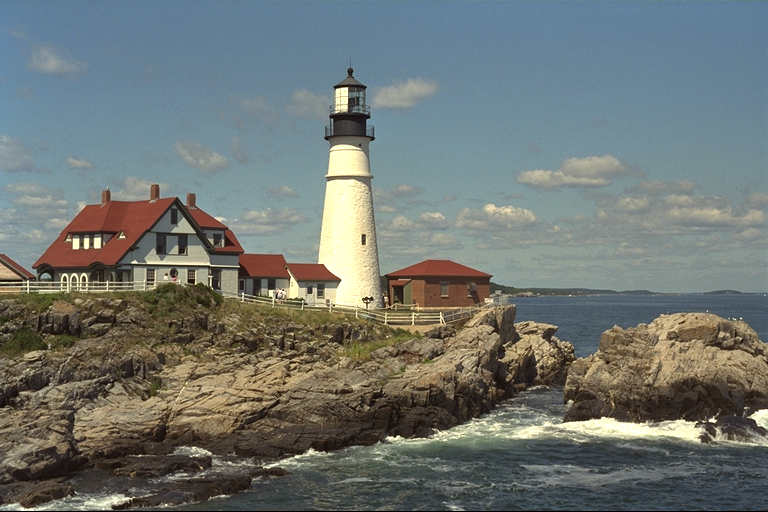}}
\subfloat[$\lambda,\tau=0.03,1$]{
\label{fig:k21_1}
\includegraphics[width=0.45\linewidth]{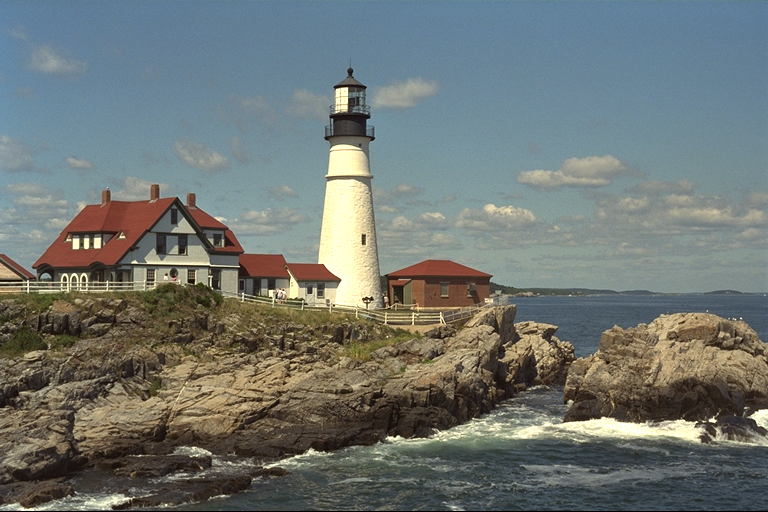}}\\
\subfloat[$\lambda,\tau=0.03,2$]{
\label{fig:k21_2}
\includegraphics[width=0.45\linewidth]{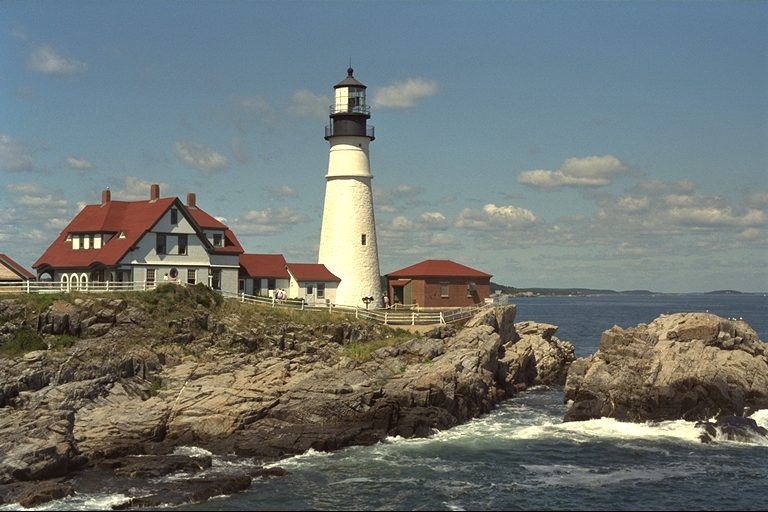}}
\subfloat[$\lambda,\tau=0.03,3$]{
\label{fig:k21_3}
\includegraphics[width=0.45\linewidth]{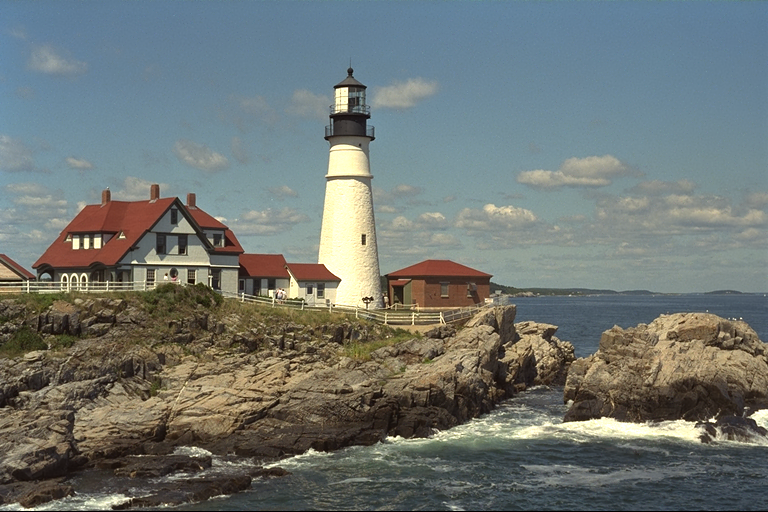}}\\
\subfloat[$\lambda,\tau=0.03,4$]{
\label{fig:k21_4}
\includegraphics[width=0.45\linewidth]{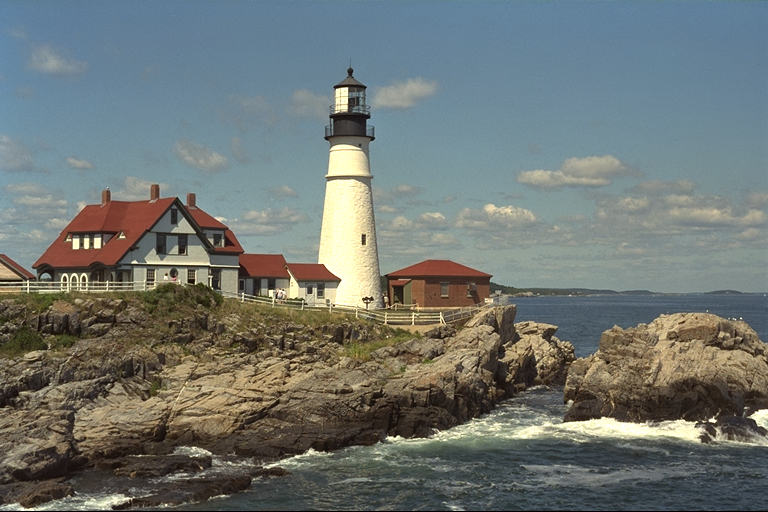}}
\subfloat[$\lambda,\tau=0.03,5$]{
\label{fig:k21_5}
\includegraphics[width=0.45\linewidth]{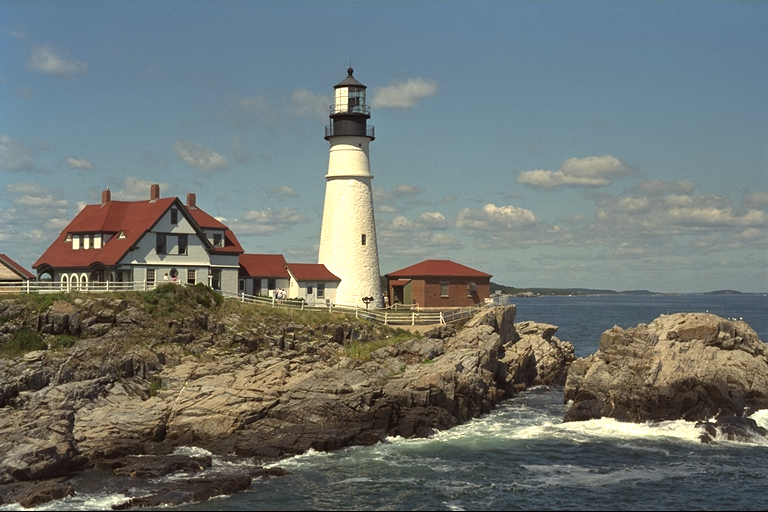}}
\end{center}
\caption{\textit{Kodim21}. Near-lossless image codec is trained with $\lambda=0.03$. Visualization of near-lossless reconstructions with different $\tau$.}
\label{fig:kodim21}
\end{figure}

}

\clearpage

\twocolumn
{\small
\bibliographystyle{ieee_fullname}
\bibliography{egbib}

\begin{thebibliography}{10}\itemsep=-1pt

\bibitem{div2k}
Eirikur Agustsson and Radu Timofte.
\newblock {NTIRE} 2017 challenge on single image super-resolution: dataset and
  study.
\newblock In {\em IEEE Conf. Comput. Vis. Pattern Recog. Worksh.}, July 2017.

\bibitem{nll1998JEI}
Rashid Ansari, Nasir~D. Memon, and Ersan Ceran.
\newblock Near-lossless image compression techniques.
\newblock {\em Journal of Electronic Imaging}, 7(3):486 -- 494, 1998.

\bibitem{Balle2017iclr}
Johannes Ball\'{e}, Valero Laparra, and Eero~P Simoncelli.
\newblock End-to-end optimized image compression.
\newblock In {\em Int. Conf. Learn. Represent.}, 2017.

\bibitem{Balle2018variational}
Johannes Ball\'{e}, David Minnen, Saurabh Singh, Sung~Jin Hwang, and Nick
  Johnston.
\newblock Variational image compression with a scale hyperprior.
\newblock In {\em Int. Conf. Learn. Represent.}, 2018.

\bibitem{bpg}
Fabrice Bellard.
\newblock {BPG} image format.
\newblock \url{https://bellard.org/bpg/}.

\bibitem{chen1994near}
Keshi Chen and Tenkasi~V Ramabadran.
\newblock Near-lossless compression of medical images through entropy-coded
  dpcm.
\newblock {\em IEEE Transactions on Medical Imaging}, 13(3):538--548, 1994.

\bibitem{cheng2020cvpr}
Zhengxue Cheng, Heming Sun, Masaru Takeuchi, and Jiro Katto.
\newblock Learned image compression with discretized gaussian mixture
  likelihoods and attention modules.
\newblock In {\em IEEE Conf. Comput. Vis. Pattern Recog.}, pages 7939--7948,
  2020.

\bibitem{Choi2019iccv}
Yoojin Choi, Mostafa El-Khamy, and Jungwon Lee.
\newblock Variable rate deep image compression with a conditional autoencoder.
\newblock In {\em Int. Conf. Comput. Vis.}, pages 3146--3154, 2019.

\bibitem{info_theory}
Thomas~M. Cover and Joy~A. Thomas.
\newblock {\em Elements of Information Theory (Wiley Series in
  Telecommunications and Signal Processing)}.
\newblock Wiley-Interscience, USA, 2006.

\bibitem{elnahas1986data}
Sharaf~E Elnahas.
\newblock Data compression with applications to digital radiology.
\newblock {\em Washington Univ., St. Louis, MO (USA)}, 1986.

\bibitem{goyal2001theoretical}
Vivek~K Goyal.
\newblock Theoretical foundations of transform coding.
\newblock {\em IEEE Sign. Process. Magazine}, 18(5):9--21, 2001.

\bibitem{guo2020cvprw}
Zongyu Guo, Yaojun Wu, Runsen Feng, Zhizheng Zhang, and Zhibo Chen.
\newblock 3-d context entropy model for improved practical image compression.
\newblock In {\em IEEE Conf. Comput. Vis. Pattern Recog. Worksh.}, pages
  116--117, 2020.

\bibitem{max2019nips}
Emiel Hoogeboom, Jorn Peters, Rianne van~den Berg, and Max Welling.
\newblock Integer discrete flows and lossless compression.
\newblock In {\em Adv. Neural Inform. Process. Syst.}, pages 12134--12144,
  2019.

\bibitem{webp}
WebP image format.
\newblock \url{https://developers.google.com/speed/webp/}.

\bibitem{ke1998near}
Ligang Ke and Michael~W Marcellin.
\newblock Near-lossless image compression: minimum-entropy, constrained-error
  dpcm.
\newblock {\em IEEE Trans. Image Process.}, 7(2):225--228, 1998.

\bibitem{kingma2015adam}
Diederik~P Kingma and Jimmy Ba.
\newblock Adam: A method for stochastic optimization.
\newblock In {\em Int. Conf. Learn. Represent.}, 2015.

\bibitem{kingma2013auto}
Diederik~P Kingma and Max Welling.
\newblock Auto-encoding variational bayes.
\newblock In {\em Int. Conf. Learn. Represent.}, 2014.

\bibitem{kingma2019introduction}
Diederik~P. Kingma and Max Welling.
\newblock An introduction to variational autoencoders.
\newblock {\em Foundations and Trends in Machine Learning}, 12(4):307--392,
  2019.

\bibitem{kingma2019bitswap}
Friso~H Kingma, Pieter Abbeel, and Jonathan Ho.
\newblock Bit-swap: Recursive bits-back coding for lossless compression with
  hierarchical latent variables.
\newblock In {\em Int. Conf. Mach. Learn.}, 2019.

\bibitem{kodak}
Eastman Kodak.
\newblock Kodak lossless true color image suite (photocd pcd0992), 1993.
\newblock \url{http://r0k.us/graphics/kodak/}.

\bibitem{krasin2017openimages}
Ivan Krasin, Tom Duerig, Neil Alldrin, Vittorio Ferrari, Sami Abu-El-Haija,
  Alina Kuznetsova, Hassan Rom, Jasper Uijlings, Stefan Popov, Andreas Veit,
  et~al.
\newblock Openimages: A public dataset for large-scale multi-label and
  multi-class image classification.
\newblock {\em Dataset available from \url{https://github. com/openimages}},
  2(3):2--3, 2017.

\bibitem{Lee2019Context}
Jooyoung Lee, Seunghyun Cho, and Seung-Kwon Beack.
\newblock Context-adaptive entropy model for end-to-end optimized image
  compression.
\newblock In {\em Int. Conf. Learn. Represent.}, May 2019.

\bibitem{li2020pami}
Mu Li, Wangmeng Zuo, Shuhang Gu, Jane You, and David Zhang.
\newblock Learning content-weighted deep image compression.
\newblock {\em IEEE Trans. Pattern Anal. Mach. Intell.}, 2020.

\bibitem{li2018cvpr}
M. Li, W. Zuo, S. Gu, D. Zhao, and D. Zhang.
\newblock Learning convolutional networks for content-weighted image
  compression.
\newblock In {\em IEEE Conf. Comput. Vis. Pattern Recog.}, pages 3214--3223,
  2018.

\bibitem{Ma2020pami}
H. Ma, D. Liu, N. Yan, H. Li, and F. Wu.
\newblock End-to-end optimized versatile image compression with wavelet-like
  transform.
\newblock {\em IEEE Trans. Pattern Anal. Mach. Intell.}, pages 1--1, 2020.

\bibitem{melnychuck1989survey}
Paul~W Melnychuck and Majid Rabbani.
\newblock Survey of lossless image coding techniques.
\newblock In {\em Digital Image Processing Applications}, volume 1075, pages
  92--100. International Society for Optics and Photonics, 1989.

\bibitem{Mentzer2018cvpr}
F. Mentzer, E. Agustsson, M. Tschannen, R. Timofte, and L.~V. Gool.
\newblock Conditional probability models for deep image compression.
\newblock In {\em IEEE Conf. Comput. Vis. Pattern Recog.}, pages 4394--4402,
  2018.

\bibitem{Mentzer2019cvpr}
F. Mentzer, E. Agustsson, M. Tschannen, R. Timofte, and L.~Van Gool.
\newblock Practical full resolution learned lossless image compression.
\newblock In {\em IEEE Conf. Comput. Vis. Pattern Recog.}, pages 10621--10630,
  2019.

\bibitem{high2020nips}
Fabian Mentzer, George Toderici, Michael Tschannen, and Eirikur Agustsson.
\newblock High-fidelity generative image compression.
\newblock In {\em Adv. Neural Inform. Process. Syst.}, 2020.

\bibitem{mentzer2020cvpr}
Fabian Mentzer, Luc Van~Gool, and Michael Tschannen.
\newblock Learning better lossless compression using lossy compression.
\newblock In {\em IEEE Conf. Comput. Vis. Pattern Recog.}, 2020.

\bibitem{minnen2018nips}
David Minnen, Johannes Ball\'{e}, and George~D Toderici.
\newblock Joint autoregressive and hierarchical priors for learned image
  compression.
\newblock In {\em Adv. Neural Inform. Process. Syst.}, pages 10771--10780,
  2018.

\bibitem{vvc}
Jens-Rainer Ohm and Gary~J Sullivan.
\newblock Versatile video coding--towards the next generation of video
  compression.
\newblock In {\em Picture Coding Symposium}, 2018.

\bibitem{pixelrnn2016icml}
A\"{a}ron van~den Oord, Nal Kalchbrenner, and Koray Kavukcuoglu.
\newblock Pixel recurrent neural networks.
\newblock In {\em Int. Conf. Mach. Learn.}, pages 1747--1756. JMLR.org, 2016.

\bibitem{rippel2017icml}
Oren Rippel and Lubomir Bourdev.
\newblock Real-time adaptive image compression.
\newblock In {\em Int. Conf. Mach. Learn.}, volume~70, pages 2922--2930, 2017.

\bibitem{pixelcnn_pp}
Tim Salimans, Andrej Karpathy, Xi Chen, and Diederik~P Kingma.
\newblock Pixelcnn++: Improving the pixelcnn with discretized logistic mixture
  likelihood and other modifications.
\newblock In {\em Int. Conf. Learn. Represent.}, 2017.

\bibitem{shannon1948mathematical}
Claude~E Shannon.
\newblock A mathematical theory of communication.
\newblock {\em The Bell system technical journal}, 27(3):379--423, 1948.

\bibitem{skodras2001j2k}
Athanassios Skodras, Charilaos Christopoulos, and Touradj Ebrahimi.
\newblock The jpeg 2000 still image compression standard.
\newblock {\em IEEE Sign. Process. Magazine}, 18(5):36--58, 2001.

\bibitem{sneyers2016flif}
Jon Sneyers and Pieter Wuille.
\newblock Flif: Free lossless image format based on maniac compression.
\newblock In {\em IEEE Int. Conf. Image Process.}, pages 66--70. IEEE, 2016.

\bibitem{theis2017iclr}
Lucas Theis, Wenzhe Shi, Andrew Cunningham, and Ferenc Husz\'{a}r.
\newblock Lossy image compression with compressive autoencoders.
\newblock In {\em Int. Conf. Learn. Represent.}, 2017.

\bibitem{Toderici2016iclr}
George Toderici, Sean~M. O'Malley, Sung~Jin Hwang, Damien Vincent, David
  Minnen, Shumeet Baluja, Michele Covell, and Rahul Sukthankar.
\newblock Variable rate image compression with recurrent neural networks.
\newblock In {\em Int. Conf. Learn. Represent.}, 2016.

\bibitem{toderici2017full}
George Toderici, Damien Vincent, Nick Johnston, Sung Jin~Hwang, David Minnen,
  Joel Shor, and Michele Covell.
\newblock Full resolution image compression with recurrent neural networks.
\newblock In {\em IEEE Conf. Comput. Vis. Pattern Recog.}, pages 5306--5314,
  2017.

\bibitem{iclr2019bitback}
James Townsend, Tom Bird, and David Barber.
\newblock Practical lossless compression with latent variables using bits back
  coding.
\newblock In {\em Int. Conf. Learn. Represent.}, 2019.

\bibitem{pixelcnn}
A\"{a}ron van~den Oord, Nal Kalchbrenner, Lasse Espeholt, Oriol Vinyals, and
  Alex Graves.
\newblock Conditional image generation with pixelcnn decoders.
\newblock In {\em Adv. Neural Inform. Process. Syst.}, pages 4790--4798, 2016.

\bibitem{wallace1992jpeg}
Gregory~K Wallace.
\newblock The jpeg still picture compression standard.
\newblock {\em IEEE Trans. Consumer Electronics}, 38(1):xviii--xxxiv, 1992.

\bibitem{wang2004ssim}
Zhou Wang, Alan~C Bovik, Hamid~R Sheikh, and Eero~P Simoncelli.
\newblock Image quality assessment: from error visibility to structural
  similarity.
\newblock {\em IEEE Trans. Image Process.}, 13(4):600--612, 2004.

\bibitem{wang2003msssim}
Zhou Wang, Eero~P Simoncelli, and Alan~C Bovik.
\newblock Multiscale structural similarity for image quality assessment.
\newblock In {\em The Thrity-Seventh Asilomar Conference on Signals, Systems \&
  Computers, 2003}, volume~2, pages 1398--1402. Ieee, 2003.

\bibitem{weinberger2000loco}
Marcelo~J Weinberger, Gadiel Seroussi, and Guillermo Sapiro.
\newblock The loco-i lossless image compression algorithm: Principles and
  standardization into jpeg-ls.
\newblock {\em IEEE Trans. Image Process.}, 9(8):1309--1324, 2000.

\bibitem{arithmetic_coding}
Ian~H Witten, Radford~M Neal, and John~G Cleary.
\newblock Arithmetic coding for data compression.
\newblock {\em Communications of the ACM}, 30(6):520--540, 1987.

\bibitem{clic}
Workshop and Challenge on Learned Image~Compression.
\newblock \url{https://www.compression.cc/challenge/}.

\bibitem{calic}
X. Wu and N. Memon.
\newblock Context-based, adaptive, lossless image coding.
\newblock {\em IEEE Trans. Communications}, 45(4):437--444, 1997.

\bibitem{Wu2000nll}
Wu Xiaolin and P. Bao.
\newblock $l_\infty$ constrained high-fidelity image compression via adaptive
  context modeling.
\newblock {\em IEEE Trans. Image Process.}, 9(4):536--542, 2000.

\bibitem{zhang2019dcc}
X. Zhang and X. Wu.
\newblock Near-lossless $\ell_\infty$-constrained image decompression via deep
  neural network.
\newblock In {\em Data Compression Conference}, pages 33--42, 2019.

\bibitem{zhang2020ultra}
X. Zhang and X. Wu.
\newblock Ultra high fidelity deep image decompression with
  $\ell_\infty$-constrained compression.
\newblock {\em IEEE Trans. Image Process.}, 30:963--975, 2021.

\end{thebibliography}
}

\end{document}